%% file: steinermasterfile.tex
\documentclass[a4paper,UKenglish]{lipics}

\usepackage{amsfonts}
\usepackage{graphicx}
\RequirePackage{fancyhdr}
\usepackage{xcolor}
\usepackage{boxedminipage}

\usepackage{vmargin}
\setmarginsrb{1.1in}{1.1in}{1.1in}{1.1in}{0mm}{0mm}{0mm}{7mm}
\usepackage{amssymb,amsmath,amsthm}

\usepackage[linesnumbered, boxed, algosection]{algorithm2e}

\usepackage{amstext}

\newtheorem{proposition}{Proposition}
\newcommand{\defparproblem}[4]{
  \vspace{1mm}
\noindent\fbox{
  \begin{minipage}{0.96\textwidth}
  \begin{tabular*}{\textwidth}{@{\extracolsep{\fill}}lr} #1  & {\bf{Parameter:}} #3 \\ \end{tabular*}
  {\bf{Input:}} #2  \\
  {\bf{Question:}} #4
  \end{minipage}
  }
  \vspace{1mm}
}

\title{Parameterized Complexity of Directed Steiner Tree on Sparse Graphs}
\titlerunning{Parameterized Complexity of DST on Sparse Graphs} 

\author[1]{Mark Jones}
\author[2]{Daniel Lokshtanov}
\author[3]{M.S. Ramanujan}
\author[3]{Saket Saurabh}
\author[4]{Ond\v rej Such\' y\footnote{Main work was
    done while with the Saarland University, supported by the DFG
    Cluster of Excellence 
    MMCI and the DFG project DARE (GU 1023/1-2), while at TU Berlin, supported by the DFG project AREG (NI 369/9), and while
    visiting IMSC Chennai supported by IMPECS.}}
\affil[1]{Department of Computer Science, Royal Holloway University of London, United Kingdom\\
  \texttt{markj@cs.rhul.ac.uk}}
  \affil[2]{University of California, San Diego\\
  
   \texttt{daniello@ii.uib.no}  }
\affil[3]{The Institute of Mathematical Sciences, \\ Chennai 600113, India. \\
 \texttt{\{msramanujan|saket\}@imsc.res.in} }
 \affil[4]{Faculty of Information Technology, Czech Technical University in Prague, Czech Republic\\
 \texttt{ondrej.suchy@fit.cvut.cz}
 }
\authorrunning{M. Jones, D. Lokshtanov, M.S. Ramanujan, S. Saurabh, and O. Such\' y} 

\Copyright[nc-nd]
          {M. Jones, D. Lokshtanov, M.S. Ramanujan, S. Saurabh, and O. Such\' y}


\serieslogo{}
\volumeinfo
  {}
  {2}
  {}
  {1}
  {1}
  {1}
\EventShortName{}
\DOI{10.4230/LIPIcs.xxx.yyy.p}

\newcommand{\N}{{\mathbb{N}}}
\newcommand{\bigoh}{O}
\newcommand{\FPT}{\textsf{FPT}}
\newcommand{\NPH}{\textsf{NP}-hard}

\newcommand{\ETH}{\textsf{ETH}}
\newcommand{\Yes}{{\sc Yes}}
\newcommand{\No}{{\sc No}}

\newcommand{\scite}[1]{{\sc (\cite{#1})}}
\newcommand{\dstfull}{{\sc Directed Steiner Tree}}
\newcommand{\dst}{{\sc DST}}
\newcommand{\DST}{{\sc DST}}

\newcommand{\stree}{{\sc Steiner Tree}}

\newcommand{\ds}{{\sc Dominating Set}}

\newtheorem{observation}[theorem]{Observation}
\newtheorem{rrule}{Rule}
\newcommand{\decnamedefn}[3]{
  \begin{tabbing} #1\\

    \emph{Input:} \hspace{1.2cm} \= \parbox[t]{12cm}{#2} \\
    \emph{Question:}             \> \parbox[t]{12cm}{#3} \\
  \end{tabbing}
}

\begin{document}

\maketitle

\begin{abstract}
\input{abstract.tex}
\end{abstract}
\setcounter{footnote}{0}

\section{Introduction}

\input{introduction.tex}

\section{Preliminaries}

\input{prelim.tex}

\section{{\dst} on sparse graphs}
In this section, we introduce our main idea and use it to design algorithms for the {\dstfull} problem on classes of sparse graphs. We begin by giving a $\bigoh^*(2^{\bigoh(hk)})$ algorithm for {\dst} on $K_h$-minor free graphs. Following that, we give a $\bigoh^*(f(h)^{k})$ algorithm for {\dst} on $K_h$-topological minor free graphs for some $f$. Then, we show that in general, even in 2 degenerated graphs, we cannot expect to have an {\FPT} algorithm for {\dst} parameterized by the solution size. Finally, we show that when the graph induced on the terminals is acyclic, then our ideas are applicable and we can give a $\bigoh^*(2^{\bigoh(hk)})$ algorithm on $K_h$-topological minor free graphs and a $\bigoh^*(2^{\bigoh(dk)})$ algorithm on $d$-degenerated graphs.

\subsection{{\dst} on minor free graphs}\label{sec:minor_free}

\input{dsthminor.tex}

\subsection{{\dst} on graphs excluding topological minors}

\input{dsttopominor.tex}

\subsection{{\dst} on $d$-degenerated graphs}

\input{dstdegenerate.tex}

\noindent
Before concluding this section, we also observe that analogous to the algorithms in Theorems \ref{thm:minorfree_algo} and \ref{thm:d_deg_algo}, we can show that in the case when the digraph induced by terminals is acyclic, the {\dst} problem admits an algorithm running in time $\bigoh^*(2^{\bigoh(hk)})$ on graphs excluding $K_h$ as a topological minor.

\begin{theorem}\label{thm:top_algo}
 {\dst} can be solved in time $\bigoh^*(2^{\bigoh(hk)})$ on graphs excluding $K_h$ as a topological minor if the digraph induced by terminals is acyclic.
\end{theorem}

\noindent
Combined with Lemma~\ref{lem:littleo}, Theorem~\ref{thm:top_algo} has the following corollary. 

\begin{corollary}
If $\mathcal{C}$ is a class of digraphs excluding $o(\log n)$-sized topological minors, then {\dst} parameterized by $k$ is {\FPT} on $\mathcal{C}$ if the digraph induced by terminals is acyclic. 
\end{corollary}

%
\subsection{Hardness of {\DST}}
\input{domsethardness.tex}

\section{Applications to {\ds}}
\label{section:domset}
\input{domsetalgo.tex}
\section{Conclusions}
\input{conclusion.tex}

\bibliographystyle{abbrv}
\bibliography{references}

%
%


\end{document}

%% file: abstract.tex
We study the parameterized complexity of the directed variant of the classical {\sc Steiner Tree} problem on various classes of directed sparse graphs. While the parameterized complexity of {\sc Steiner Tree} parameterized by the number of terminals is well understood, not much is known about the parameterization by the number of non-terminals in the solution tree. All that is known for this parameterization is that both the directed and the undirected versions are W[2]-hard on general graphs, and hence unlikely to be fixed parameter tractable (\FPT{}). The undirected {\sc Steiner Tree} problem becomes \FPT{} when restricted to sparse classes of graphs such as planar graphs, but the techniques used to show this result break down on directed planar graphs. 

In this article we precisely chart the tractability border for {\sc Directed Steiner Tree} (\DST{}) on sparse graphs parameterized by the number of non-terminals in the solution tree. Specifically, we show that the problem is fixed parameter tractable on graphs excluding a topological minor, but becomes W[2]-hard on graphs of degeneracy 2. On the other hand we show that if the subgraph induced by the terminals is required to be acyclic then the problem becomes \FPT{} on graphs of bounded degeneracy. 


We further show that our algorithm achieves the best possible running time dependence on the solution size and degeneracy of the input graph, under standard complexity theoretic assumptions. Using the ideas developed for \DST{}, we also obtain improved algorithms for {\sc Dominating Set}  on 
sparse undirected graphs. These algorithms are asymptotically optimal. 

%% file: introduction.tex
In the {\sc Steiner Tree} problem we are given as input a $n$-vertex graph $G=(V,E)$ and a set $T \subseteq V$ of terminals. The objective is to find a subtree $ST$ of $G$ spanning $T$ that minimizes the number of vertices in $ST$. {\sc Steiner Tree} is one of the most intensively studied graph problems in Computer Science.
%
%
Steiner trees are important in various applications such as VLSI routings~\cite{KahngRobins95}, phylogenetic tree reconstruction~\cite{HwangRichardsWinter92} and network routing~\cite{KortePromelSteger90}. We refer to the book of Pr{\"o}mel and Steger \cite{PromelS02} for an overview of the results on, and applications of the {\sc Steiner Tree} problem. The {\sc Steiner Tree} problem is known to be NP-hard~\cite{GareyJohnson79}, and remains hard even on planar graphs~\cite{GareyJ77}. The minimum number of non-terminals can be approximated to within $O(\log n)$, but cannot be approximated to $o(\log t)$, where $t$ is the number of terminals, unless P $\subseteq$ DTIME[$n^{\text{polylog } n}$] (see~\cite{KleinR95}). Furthermore the weighted variant of {\sc Steiner Tree} remains APX-complete, even when the graph is complete and all edge costs are either $1$ or $2$ (see~\cite{BernPlassmann89ipl}). 

In this paper we study a natural generalization of {\sc Steiner Tree} to directed graphs, from the perspective of parameterized complexity. 
%
The goal of parameterized complexity is to find ways of solving {\NPH} problems more efficiently than by brute force. The aim is to restrict the combinatorial explosion in the running time to a parameter that is much smaller than the input size for many input instances occurring in practice. Formally, a {\em parameterization} of a problem is the assignment of an integer $k$ to each input instance and we say that a parameterized problem is {\em fixed-parameter tractable} ({\FPT}) if there is an algorithm that solves the problem in time $f(k)\cdot |I|^{\bigoh(1)}$, where $|I|$ is the size of the input instance and $f$ is an arbitrary computable function depending only on the parameter $k$. 
Above {\FPT},  there exists a hierarchy of complexity classes, known as the  W-hierarchy.  Just as NP-hardness is used as an evidence that a problem is probably not polynomial time solvable, showing that a parameterized problem is hard for one of these classes gives evidence that the problem is unlikely to be fixed-parameter tractable. The main classes in this hierarchy are:
$$ \mbox{FPT $ \subseteq$ W[1] $ \subseteq$ W[2]} \subseteq \cdots \subseteq \mbox{W[P] $\subseteq$ XP}$$
The principal analogue of the classical intractability class NP is W[1].  In particular, this means that an FPT algorithm for any W[1]-hard problem would yield a $O(f(k)n^c)$ time algorithm for every problem in the class 
W[1]. 
$X$P is the class of all problems that are solvable in time $O(n^{g(k)})$. Here, $g$ is some (usually computable) function. For more background on parameterized complexity the reader is referred to the monographs \cite{DF99,FG06,Nie06}. We consider the following directed variant of {\sc Steiner Tree}.

\defparproblem{\sc Directed Steiner Tree (DST)} 
{A directed graph $D=(V,A)$, a root vertex $r \in V$, a set $T \subseteq V \setminus \{r\}$ of terminals and an integer $k\in \N$.} {$k$}{Is there a set $S \subseteq V \setminus (T \cup \{r\})$ of at most $k$ vertices such that the digraph $D[S \cup T \cup \{r\}]$ contains a directed path from $r$ to every terminal $t \in T$?}

The {\dst} problem is well studied in approximation algorithms, as the problem generalizes several important connectivity and domination problems on undirected as well as directed graphs~\cite{CharikarCCDGGL99,DemaineHK09a,GuhaK98,HalperinKKSW07,Zelikovsky97,ZosinK02}. These include {\sc Group Steiner Tree}, {\sc Node Weighted Steiner Tree}, {\sc TSP} and {\sc Connected Dominating Set}. However, this problem has so far largely been ignored in the realm of parameterized complexity. The aim of this paper is to fill this gap.

It follows from the reduction presented in~\cite{MolleRR08} that {\dst} is W[2]-hard on general digraphs. Hence we do not expect {\FPT} algorithms to exist for these problems, and so we turn our attention to classes of \emph{sparse} digraphs. Our results give a nearly complete picture of the parameterized complexity of \DST{} on sparse digraphs. Specifically, we prove the following results. We use the $\bigoh^*$ notation to suppress factors polynomial in the input size.
\begin{enumerate}
\item There is a $\bigoh^*(2^{\bigoh(hk)})$-time algorithm for \DST{} on digraphs excluding $K_h$ as a minor\footnote{When we say that a digraph excludes a fixed (undirected) graph as a minor or a topological minor, or that the digraph has degeneracy $d$ we mean that the statement is true for the underlying undirected graph.}. Here $K_h$ is a clique on $h$ vertices.
\item There is a $\bigoh^*(f(h)^k)$-time algorithm for \DST{} on digraphs excluding $K_h$ as a topological minor.
\item There is a $\bigoh^*(2^{\bigoh(hk)})$-time algorithm for \DST{} on digraphs excluding $K_h$ as a topological minor if the graph induced on terminals is acyclic.
\item {\dst} is W[2]-hard on 2-degenerated digraphs if the graph induced on terminals is allowed to contain directed cycles.
\item There is a $\bigoh^*(2^{\bigoh(dk)})$-time algorithm for \DST{} on $d$-degenerated graphs if the graph induced on terminals is acyclic, implying that {\dst} is {\FPT} parameterized by $k$ on $o(\log n)$-degenerated graph classes. This yields the first {\FPT} algorithm for {\sc Steiner Tree} on {\em undirected} $d$-degenerate graphs.
\item For any constant $c>0$, there is no $f(k)n^{o({\frac k {\log k} })}$-time algorithm on graphs of degeneracy $c \log n$ even if the graph induced on terminals is acyclic, unless the Exponential Time Hypothesis~\cite{ImpagliazzoPZ01} ({\ETH}) fails.
\end{enumerate}


Our algorithms for {\dst} hinge on a novel branching which exploits the domination-like nature of the {\dst} problem. The branching is based on a new measure which seems useful for various connectivity and domination problems on both directed and undirected graphs of bounded degeneracy.  We demonstrate the versatility of the new branching by applying it to the {\ds} problem on graphs excluding a topological minor and more generally, graphs of bounded degeneracy. The well-known {\ds} problem is defined as follows.

\defparproblem{\sc  Dominating Set} 
{An undirected graph $G=(V,E)$, and an integer $k\in \N$.} {$k$}{Is there a set $S \subseteq V$ of at most $k$ vertices such that every vertex in $G$ is either in $S$ or adjacent to a vertex in $S$?}

\noindent
Our  $\bigoh^*(2^{\bigoh(dk)})$-time algorithm  for {\ds} on $d$-degenerated graphs improves over the $O^*(k^{O(dk)})$ time algorithm by Alon and Gutner~\cite{AlonG09}. It turns out that our algorithm is essentially optimal -- we show that assuming the \ETH{}, the running time dependence of our algorithm on the degeneracy of the input graph and solution size $k$ can not be significantly improved.  Using these ideas we also obtain a polynomial time $O(d^2)$ factor 
approximation algorithm for {\sc  Dominating Set}  on $d$-degenerate graphs. We give survey of existing literature on {\sc Dominating Set} and the results for it in Section~\ref{section:domset}. 
%
We believe that our new branching and corresponding measure will turn out to be useful for several other problems on sparse (di)graphs.
%
%
\begin{figure}[t]
 \centering
 \includegraphics[width=280 pt,height=155 pt]{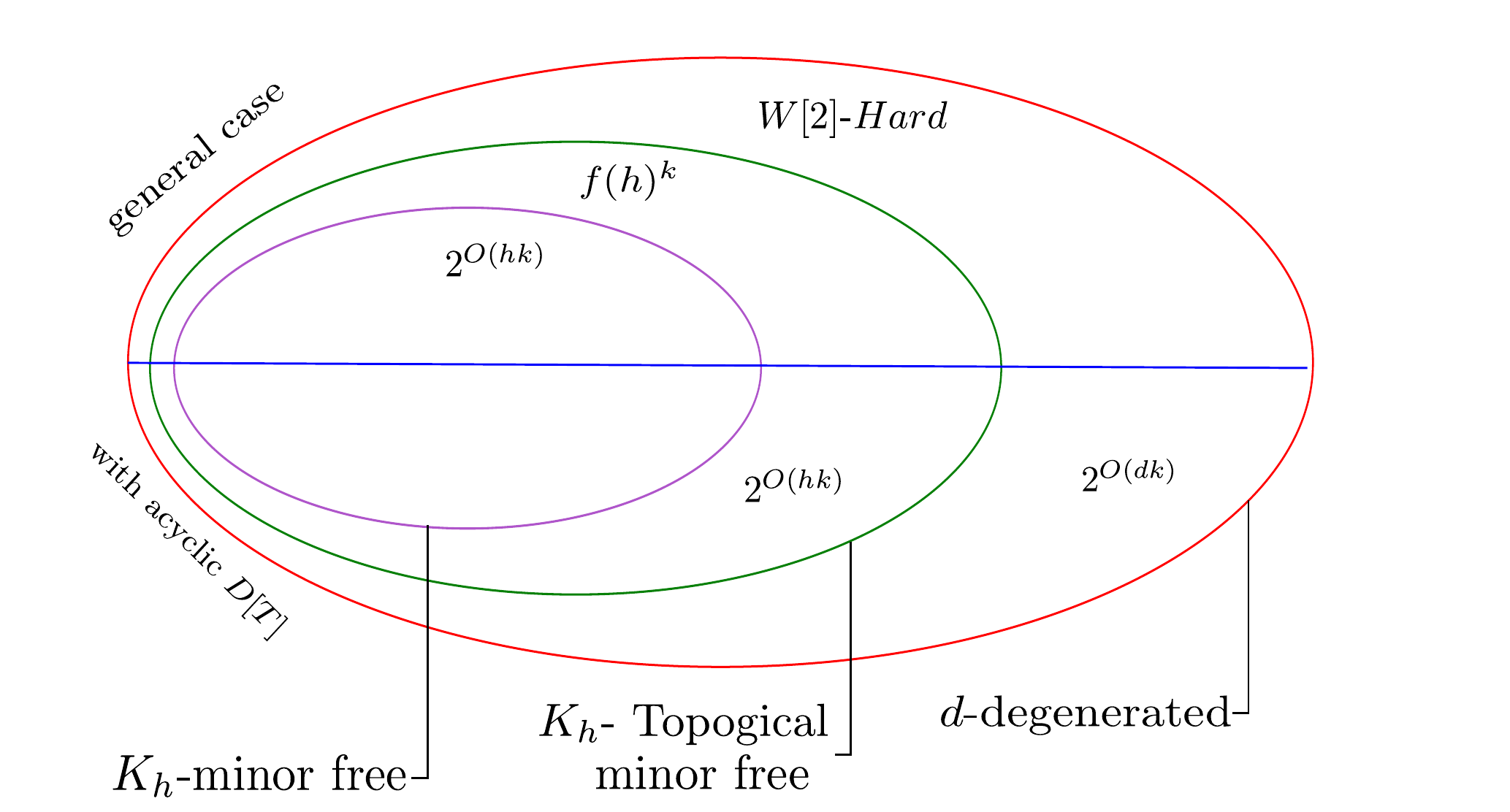}
\caption{A summary of the main results in the paper}
\label{fig:hierarchy}
\end{figure}

\medskip

\noindent
{\bf Related Results.} Though the parameterized complexity of \DST{} has so far been largely ignored, it has not been left completely unexplored. In particular the classical dynamic programming algorithm by Dreyfus and Wagner~\cite{DreyfusWagner72networks} from $1972$ solves {\sc Steiner Tree} in time $O^*(3^t)$ where $t$ is the number of terminals in the input graph. The algorithm can also be used to solve \DST{} within the same running time, and may be viewed as a FPT algorithm for {\sc Steiner Tree} and \DST{} if the number of terminals in the instance is the parameter. Fuchs et al.~\cite{FuchsKMRRW07} improved the  algorithm of Dreyfus and Wagner and obtained an algorithm with running time $O^*((2+\epsilon)^t)$, for any constant $\epsilon >0$. More recently, Bj\"orklund, Husfeldt, Kaski, and Koivisto \cite{BjorklundHKK07} obtained an $O^*(2^t)$ time algorithm for the cardinality version of {\sc Steiner Tree}.  Finally, Nederlof ~\cite{Nederlof09} obtained an 
algorithm running in $O^*(2^t)$ and polynomial space. All of these algorithms can also be modified to work for \DST{}.

For most hard problems, the most frequently studied parameter in parameterized complexity is the size or quality of the solution.  For {\sc Steiner Tree} and \DST{}, however, this is not the case. The non-standard parameterization of the problem by the number of terminals is well-studied, while the standard parameterization by the number of non-terminals in the solution tree has been left unexplored, aside from the simple W[2]-hardness proofs~\cite{MolleRR08}. Steiner-type problems in directed graphs from parameterized perspective were studied in~\cite{GuoNS11} in arc-weighted setting, but the paper focuses more on problems in which the required connectivity among the terminals is more complicated than just a tree.

For {\sc Steiner Tree} parameterized by the solution size $k$, there is a simple (folklore) FPT algorithm on planar graphs. The algorithm is based on the fact that planar graphs have the diameter-treewidth property~\cite{Eppstein00}, the fact that {\stree} can be solved in polynomial time on graphs of bounded treewidth~\cite{CyganNPPRW11} along with a simple preprocessing step. In this step, one contracts adjacent terminals to single vertices and removes all vertices at distance at least $k+1$ from any terminal. For \DST{}, however, this preprocessing 
step breaks down. Thus, previous to this work, nothing is known about the standard parameterization of \DST{} aside from the W[2]-hardness result on general graphs.

%% file: prelim.tex
Given a digraph $D=(V,A)$, for each vertex $v \in V$,  
we define $N^{+}(v) = \{ w \in V \vert (v,w) \in A \}$ and $N^{-}(v) = \{w \in H \vert (w,v) \in A  \}$. In other words, the sets $N^{+}(v)$ and $N^{-}(v)$ are the set of out-neighbors and in-neighbors of $v$, respectively.

Degeneracy of an undirected graph $G=(V,E)$ is defined as the least number $d$ such that every subgraph of $G$ contains a vertex of degree at most $d$. Degeneracy of a digraph is defined to be the degeneracy of the underlying undirected graph.
We say that a class of (di)graphs $\mathcal{C}$ is \emph{$o(\log n)$-degenerated} if there is a function $f(n) = o(\log n)$ such that every (di)graph $G \in \mathcal{C}$ is $f(|V(G)|)$-degenerated.

In a directed graph, we say that a vertex $u$ {\em dominates} a vertex $v$ if there is an arc $(u,v)$ and in an undirected graph, we say that a vertex $u$ dominates a vertex $v$ if there is an edge $(u,v)$ in the graph.

Given a vertex $v$ in a directed graph $D$, we define the operation of \emph{short-circuiting} across $v$ as follows. We add an arc from every vertex in $N^-(v)$ to every vertex in $N^+(v)$ and delete~$v$. 

For a set of vertices $X \subseteq V(G)$ such that $G[X]$ is connected we denote by $G/X$ the graph obtained by contracting edges of a spanning tree of $G[X]$ in $G$.

Given an instance $(D,r,T,k)$ of {\dst}, we say that 
a set $S \subseteq V \setminus (T \cup \{r\})$ of at most $k$ vertices is \emph{a solution} to this instance if in the digraph $D[S \cup T \cup \{r\}]$ there is a directed path from $r$ to every terminal $t \in T$ .

\noindent
{\bf Minors and Topological Minors.} For a graph $G=(V,E)$, a graph $H$ is a {\em minor} of $G$ if $H$ can be obtained from $G$ by deleting vertices, deleting edges, and contracting edges. We denote that $H$ is a minor of $G$ by $H\preceq G$. A mapping $\varphi: V(H) \to 2^{V(G)}$ is a model of $H$ in $G$ if for every $u,v \in V(H)$ with $u \neq v$ we have $\varphi(u) \cap \varphi(v) = \emptyset$, $G[\varphi(u)]$ is connected, and, if $\{u,v\}$ is an edge of $H$, then there are $u' \in \varphi(u)$ and $v' \in \varphi(v)$ such that $\{u',v'\} \in E(G)$. It is known, that $H \preceq G$ iff $H$ has a model in $G$. 

A \emph{subdivision} of a graph $H$ is obtained by replacing each edge of $H$ by a non-trivial path. We say that $H$ is a \emph{topological minor} of $G$ if some subgraph of $G$ is isomorphic to a subdivision of $H$ and denote it by $H\preceq_T G$. In this paper, whenever we make a statement about a directed graph having (or being) a minor of another graph, we mean the underlying undirected graph. 
A graph $G$ \emph{excludes graph $H$ as a (topological) minor} if $H$ is not a (topological) minor of $G$.
We say that a class of graphs $\mathcal{C}$ \emph{excludes $o(\log n)$-sized (topological) minors} if there is a function $f(n) = o(\log n)$ such that for every graph $G \in \mathcal{C}$ we have that $K_{f(|V(G)|)}$ is not a (topological) minor of $G$.

\noindent
{\bf Tree Decompositions.} A \emph{tree decomposition} of a graph $G=(V,E)$ is a pair $(M,\beta)$ where $M$ is a rooted tree and $\beta:V(M)\rightarrow 2^V$, such that :

\begin{enumerate}
\item $\bigcup_{t\in V(M)}\beta(t)=V$.
\item For each edge $(u,v)\in E$, there is a $t\in V(M)$ such that both $u$ and $v$ belong to $\beta(t)$.
\item For each $v\in V$, the nodes in the set $\{t\in V(M)\mid v\in \beta(t)\}$ form a connected subtree of $M$.
\end{enumerate}

\noindent
The following notations are the same as that in \cite{GroheM12}. Given a tree decomposition of graph $G=(V,E)$, we define mappings $\sigma, \gamma, \alpha:V(M)\rightarrow2^V$ by letting for all $t\in V(M)$,\\

\begin{center}
$
\sigma(t) = \begin{cases}

  \emptyset & \text{if $t$ is the root of $M$} \\

  \beta(t)\cap \beta(s) & \text{if $s$ is the parent of $t$ in $M$} \\

\end{cases}
$\end{center}

\hspace{87 pt}$\gamma(t)=\bigcup_{u \text { is a descendant of } t}\beta(u)$\\

\hspace{87 pt}$\alpha(t)=\gamma(t)\setminus \sigma(t)$.\\

\noindent
Let $(M,\beta)$ be a tree decomposition of a graph $G$. The {\em width} of $(M,\beta)$ is $min\{\vert \beta(t)\vert -1 \mid t\in V(M)$\}, and the {\em adhesion}  of the tree decomposition is $max\{\vert \sigma(t)\vert \mid t\in V(M)\}$. For every node $t\in V(M)$, the {\em torso} at $t$ is the graph \\

\begin{center}
$\tau(t):= G[\beta(t)] \cup E(K[\sigma(t)]) \cup \bigcup_{u\, \mathrm{ child}\, \mathrm{of }\, t} E(K[\sigma(u)])$.
\end{center}

Again, by a tree decomposition of a directed graph, we mean a tree decomposition for the underlying undirected graph.

%% file: dsthminor.tex
\noindent
We begin with a polynomial time preprocessing which will allow us to identify a \emph{special} subset of the terminals with the property that it is enough for us to find an arborescence from the root to these terminals.

\begin{rrule}\label{rul:scc}
Given an instance $(D,r,T,k)$ of {\dst}, let $C$ be a strongly connected component with at least 2 vertices in the graph $D[T]$. Then, contract $C$ to a single vertex $c$, to obtain the graph $D^\prime$ and return the instance $(D^\prime,r,T'= (T\setminus C) \cup \{c\},k)$. 
\end{rrule}

\noindent
{\bf Correctness.} Suppose $S$ is a solution to $(D,r,T,k)$. Then there is a directed path from $r$ to every terminal $t \in T$ in the digraph $D[S \cup T \cup \{r\}]$. Contracting the vertices of $C$ will preserve this path.
Hence, $S$ is also a solution for $(D^\prime,r,T',k)$.

Conversely, suppose $S$ is a solution for $(D^\prime,r,T',k)$. If the path $P$ from $r$ to some $t \in T'\setminus C$ in $D'[S \cup T' \cup \{r\}]$ contains $c$, then there must be a path from $r$ to some vertex $x$ of $C$ and a path (possibly trivial) from some vertex $y \in C$ to $t$ in $D[S \cup T \cup \{r\}]$. As there is a path between any $x$ and $y$ in $D[C]$, concatenating these three paths results in a path from $r$ to $t$ in $D[S \cup T \cup \{r\}]$. Hence, $S$ is also a solution to $(D,r,T,k)$.


%

\begin{proposition}\label{prop:contraction}
 Given an undirected graph $G=(V,E)$ which excludes $K_h$ as a minor for some $h$, and a vertex subset $X\subseteq V$ inducing a connected subgraph of $G$, the graph $G/X$ also excludes $K_h$ as a minor.
\end{proposition}

\noindent
We call an instance \emph{reduced} if Rule~\ref{rul:scc} cannot be applied to it.
Given an instance $(D,r,T,k)$, we first apply Rule~\ref{rul:scc} exhaustively to obtain a reduced instance. Since the resulting graph still excludes $K_h$ as a minor (by Proposition~\ref{prop:contraction}), we have not changed the problem and hence, for ease of presentation, we denote the reduced instance also by $(D,r,T,k)$.
 We call a terminal vertex $t \in T$ a \emph{source-terminal} if it has no in-neighbors in $D[T]$. We use $T_0$ to denote the set of all source-terminals. Since for every terminal, the graph $D[T]$ contains a path from some source terminal to this terminal, we have the following observation.

\begin{observation}\label{obs:t0}
 Let $(D,r,T,k)$ be a reduced instance and let $S \subseteq V$. Then the digraph $D[S \cup T \cup \{r\}]$ contains a directed path from $r$ to every terminal $t \in T$ if and only if it contains a directed path from $r$ to every source-terminal $t \in T_0$.
\end{observation}

The following is an important subroutine of our algorithm.

\begin{lemma}\label{lem:nederlof}
 Let $D$ be a digraph, $r \in V(D)$, $T \subseteq V(D) \setminus \{r\}$ and $T_0 \subseteq T$. There is an algorithm which can find a minimum size set $S \subseteq V(D)$ such that there is path from $r$ to every $t \in T_0$ in $D[T \cup \{r\} \cup S]$ in time $\bigoh^*(2^{|T_0|})$.
\end{lemma}

\begin{proof}
 Nederlof~\cite{Nederlof09} gave an algorithm to solve the {\sc Steiner Tree} problem on undirected graphs in time $\bigoh^*(2^t)$ where $t$ is the number of terminals. Misra et al.~\cite{MisraPRSS10} observed that the same algorithm can be easily modified to solve the {\dst} problem in time $\bigoh^*(2^t)$ with $t$ being the number of terminals.  In our case, we create an instance of the {\dst} problem by taking the same graph, defining the set of terminals as $T_0$ and for every vertex $t\in T\setminus T_0$, \emph{short-circuiting} across this vertex. Clearly, a $k$-sized solution to this instance gives a $k$-sized solution to the original problem. 
 To actually find the set of minimum size, we can first find its size by a binary search and then delete one by one the non-terminals, if their deletion does not increase the size of the minimum solution.
\end{proof}

\noindent
We call the algorithm from Lemma~\ref{lem:nederlof}, \textsc{Nederlof}$(D,r,T,T_0)$.

\noindent
We also need the following structural claim regarding the existence of low degree vertices in graphs excluding $K_h$ as a topological  minor.

\begin{lemma}\label{lem:minorfree_small_degree}
 Let $G=(V,E)$ be an undirected graph excluding $K_h$ as a topological minor and let $X,Y\subseteq V$ be two disjoint vertex sets. If every vertex in $X$ has at least $h-1$ neighbors in $Y$, then there is a vertex in $Y$ with at most $ch^4$ neighbors in $X\cup Y$ for some constant $c$. 
\end{lemma}

\begin{proof}
It was proved in \cite{BollobasBT98,KomloS96}, that there is a constant $a$ such that any graph that does not contain $K_h$ as a topological minor is $d=ah^2$-degenerated.
Consider the graph $H_0 = G[X\cup Y]\setminus E(X)$. We construct a sequence of graphs $H_0,\dots, H_l$, starting from $H_0$ and repeating an operation 
which ensures that any graph in the sequence excludes $K_h$ as a topological minor. The operation is defined as follows. In graph $H_i$, pick a vertex $x \in X$. As it has degree at least $h-1$ in $Y$ and there is no $K_h$ topological minor in $H_i$, it has two neighbors $y_1$ and $y_2$ in $Y$, which are non-adjacent. Remove $x$ from $H$ and add the edge $(y_1,y_2)$ to obtain the graph $H_{i+1}$. By repeating this operation, we finally obtain a graph $H_l$ where the set $X$ is empty. As the graph $H_l$ still excludes $K_h$ as a topological minor, it is $d$-degenerated, and hence it has at most $d|Y|$ edges. In the sequence of operations, every time we remove a vertex from $X$, we added an edge between two vertices of $Y$. Hence, the number of vertices in $X$ in $H_0$ is bounded by the number of edges within $Y$ in $H_l$, which is at most $d|Y|$. As $H_0$ is also $d$-degenerated, it has at most $d(|X|+|Y|)= d(d+1)|Y|$ edges. Therefore, there is a vertex in $Y$ incident on at most $2d(d+1)= 2ah^2(ah^2+1)\leq 
ch^4$ edges where $c=4a^{2}$. This concludes the proof of the lemma.
\end{proof}

\noindent
The following proposition allows us to apply Lemma~\ref{lem:minorfree_small_degree} in the case of graphs excluding $K_h$ as a minor.

\begin{proposition}\label{prop:minor-topminor}
If a graph $G$ exludes $K_h$ as a minor, it also excludes $K_h$ as a topological minor.
\end{proposition}

Let $(D,r,T,k)$ be a reduced instance of {\dst}, $Y\subseteq V\setminus T$ be a set of non-terminals representing a partial solution and $d_b$ be some fixed positive integer. We define the following sets of vertices (see Fig.~\ref{fig:partition}).
\begin{itemize}
\item $T_1 = T_1(Y)$ is the set of source terminals dominated by $Y$.
\item $B_h = B_h(Y,d_b)$ is the set of non-terminals which dominate at least $d_b+1$ terminals in $T_0\setminus T_1$. 
\item $B_l= B_l(Y,d_b)$ is the set of non-terminals which dominate at most $d_b$ terminals in $T_0\setminus T_1$. 
\item $W_h= W_h(Y,d_b)$ is the set of terminals in $T_0\setminus T_1$ which are dominated by $B_h$. 
\item $W_l= W_l(Y,d_b) = T_0\setminus (T_1 \cup W_h)$ is the set of source terminals which are not dominated by $Y$ or $B_h$.
\end{itemize}
Note that the sets are pairwise disjoint. The constant $d_b$ is introduced to describe the algorithm in a more general way so that we can use it in further sections of the paper. Throughout this section, we will have $d_b=h-2$.

\begin{lemma}\label{lem:base_case}
 Let $(D,r,T,k)$ be a reduced instance of {\dst}, $Y\subseteq V\setminus T$, $d_b \in \N$, and $T_1$, $B_h$, $B_l$, $W_h$, and $W_l$ as defined above. If $|W_l| > d_b (k - |Y|)$, then the given instance does not admit a solution containing $Y$.
\end{lemma}
\begin{proof}
 This follows from the fact that any non-terminal from $V\setminus (B_h \cup Y)$ in the solution, which dominates a vertex in $W_l$ can dominate at most $d_b$ of these vertices. Since the solution contains at most $k-\vert Y\vert$ such non-terminals, at most $d_b(k-\vert Y\vert)$ of these vertices can be dominated. This completes the proof. 
\end{proof}

\begin{lemma}\label{lem:base_case_nederlof}
 Let $(D,r,T,k)$ be a reduced instance of {\dst}, $Y\subseteq V\setminus T$, $d_b \in \N$, and $T_1$, $B_h$, $B_l$, $W_h$, and $W_l$ as defined above. 
 If $B_h$ is empty, then there is an algorithm which can test if this instance has a solution containing $Y$ in time $\bigoh^*(2^{d_b (k -|Y|) + |Y|})$.
\end{lemma}

\begin{proof}
We use Lemma~\ref{lem:nederlof} and test whether $|$\textsc{Nederlof}($D,r,T \cup Y$,$Y \cup (T_0 \setminus T_1))| \leq k$.
 We know that $\vert Y\vert \leq k$ and, by Lemma~\ref{lem:base_case}, we can assume that $\vert T_0\setminus T_1\vert \leq d_b(k-\vert Y\vert)$. Therefore, the size of $Y \cup (T_0 \setminus T_1)$ is bounded by $|Y|+d_b(k-|Y|)$, implying that we can solve the {\dst} problem on this instance in time $\bigoh^*(2^{d_b (k -|Y|) + |Y|})$. This completes the proof of the lemma.
\end{proof}

\noindent
We now proceed to the main algorithm of this subsection.

\begin{theorem}\label{thm:minorfree_algo}
 {\dst} can be solved in time $\bigoh^*(3^{hk + o(hk)})$ on graphs excluding $K_h$ as a minor.
\end{theorem}

\begin{proof}
 Let $T_0$ be the set of source terminals of this instance. The algorithm we describe takes as input a reduced instance $(D,r,T,k)$, a vertex set $Y$ and a positive integer $d_b$ 
 and returns a smallest solution for the instance which contains $Y$ if such a solution exists. If there is no solution, then 
 the algorithm returns a dummy symbol $S_\infty$. To simplify the description, we assume that $|S_\infty| = \infty$. The algorithm is a recursive algorithm and at any stage of the recursion, the corresponding recursive step returns the smallest set found in the recursions initiated in this step. We start with $Y$ being the empty set.

\begin{algorithm}[t]
   \SetKwInOut{Input}{Input}\SetKwInOut{Output}{Output}
  \Input{An instance $(D,r,T,k)$ of {\dst}, degree bound $d_b$, set $Y$ }
  \Output{A smallest solution of size at most $k$ and containing $Y$ for the instance $(D,r,T,k)$ if it exists and $S_\infty$ otherwise}
  Compute the sets $B_h$, $B_l$, $Y$, $W_h$, $W_l$\\  
  \lIf{$\vert W_l\vert>d(k-\vert Y\vert)$}{\Return $S_\infty$}

   \ElseIf{$B_h=\emptyset$}{
      $S \leftarrow$ \textsc{Nederlof}$(D,r,T \cup Y,W_l \cup Y)$.\\
 \lIf{$|S| > k$}{$S \leftarrow S_\infty$}\\
 \Return $S$\\
 }

\Else{
$S \leftarrow S_\infty$\\
  \emph{Find vertex $v\in W_h$ with the least in-neighbors in $B_h$.}\\
\For{$u\in B_h\cap N^{-}(v)$}{ $Y' \leftarrow Y\cup \{u\}$,\\
$S'\leftarrow$ {\sc DST-solve}($(D,r,T,k),d_b, Y$).\\ 
\lIf{$|S'| < |S|$}{$S \leftarrow S'$}\\
}
$D'\leftarrow D \setminus (B_h\cap N^{-}(v))$\\
$S'\leftarrow$ {\sc DST-solve}($(D^\prime,r,T,k),d_b,Y$).\\
\lIf{$|S'| < |S|$}{$S \leftarrow S'$}\\
  \Return $S$
}
 \BlankLine
  \caption{Algorithm {\sc DST-solve} for {\dst} on graphs excluding $K_h$ as a minor }\label{algo:minor_algo}
\end{algorithm}

\begin{figure}[t]
 \centering
 \includegraphics[width=250 pt,height=155 pt]{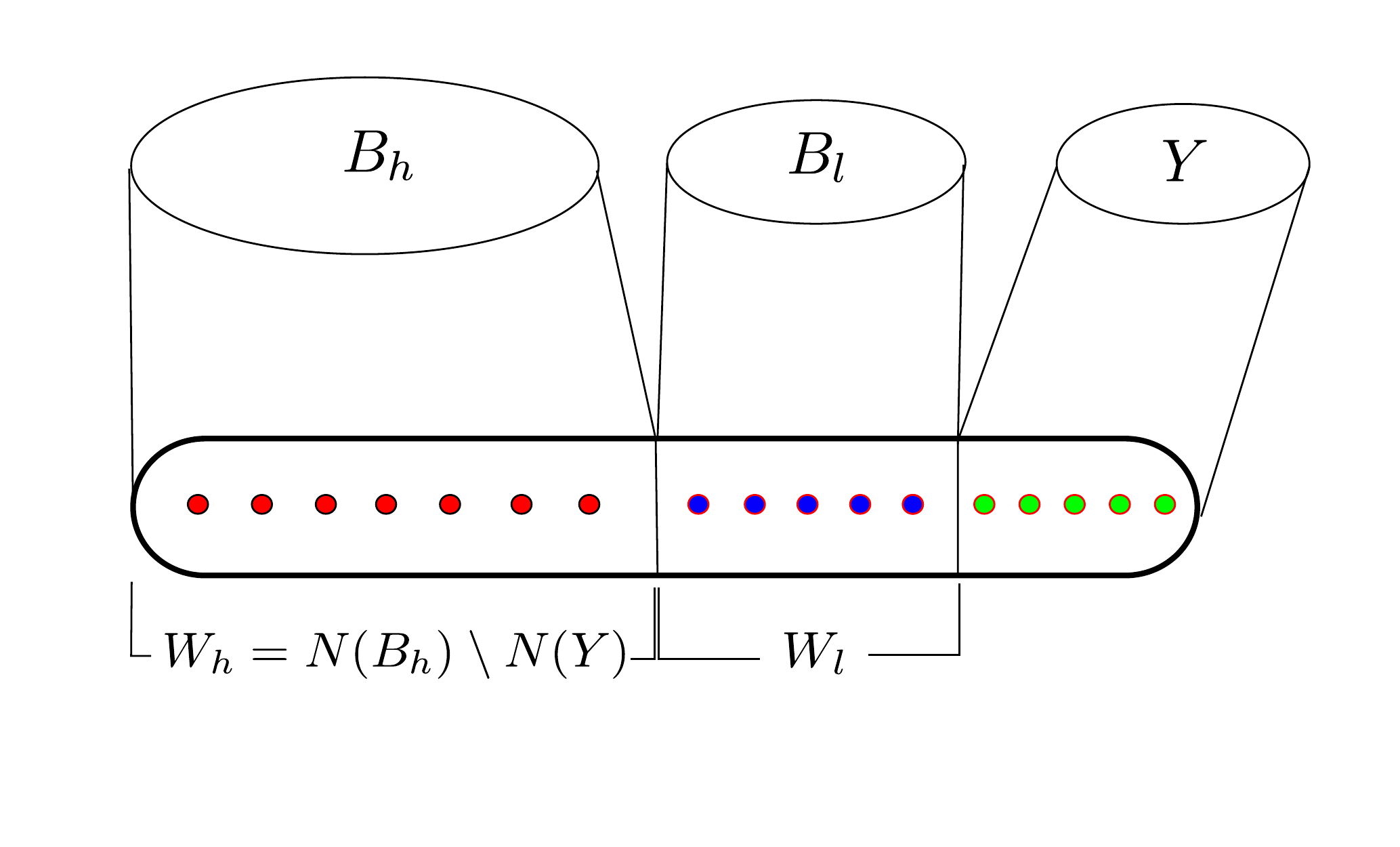}
\caption{An illustration of the sets defined in Theorem~\ref{thm:minorfree_algo}}
\label{fig:partition}
\end{figure}

By Lemma~\ref{lem:base_case}, if $\vert W_l\vert>d_b(k-\vert Y\vert)$, then there is no solution containing $Y$ and hence we return $S_\infty$ (see Algorithm~\ref{algo:minor_algo}). If $B_h$ is empty, then we apply Lemma~\ref{lem:base_case_nederlof} to solve the problem in time $\bigoh^*(2^{d_b k})$. If $B_h$ is non-empty, then we find a vertex $v\in W_h$ with the least in-neighbors in $B_h$. Suppose it has $d_w$ of them.

We then branch into $d_w+1$ branches described as follows. In the first $d_w$ branches, we move a vertex $u$ of $B_h$ which is an in-neighbor of $v$, to the set $Y$. Each of these branches is equivalent to picking one of the in-neighbors of $v$ from $B_h$ in the solution. We then recurse on the resulting instance. In the last of the $d_w+1$ branches, we delete from the instance non-terminals in $B_h$ which dominate $v$ and recurse on the resulting instance. Note that in the resulting instance of this branch, we have $v$ in $W_l(Y)$.\\

\noindent
{\bf Correctness.}
At each node of the recursion tree, we define a measure $\mu(I)=d_b (k-\vert Y\vert)-\vert W_l\vert$. We prove the correctness of the algorithm by induction on this measure. In the base case, when $d_b (k-\vert Y\vert) - \vert W_l\vert< 0$, then the algorithm is correct (by Lemma~\ref{lem:base_case}). Now, we assume as induction hypothesis that the algorithm is correct on instances with measure less than some $\mu\geq 0$. Consider an instance $I$ such that $\mu(I)=\mu$. Since the branching is exhaustive, it is sufficient to show that the algorithm is correct on each of the child instances. To show this, it is sufficient to show that for each child instance $I^\prime$, $\mu(I^\prime)<\mu(I)$. In the first $d_w$ branches, the size of the set $Y$ increases by 1, and the size of the set $W_l$ does not decrease. Hence, in each of these branches, $\mu(I^\prime)\leq \mu(I)-d_b$. In the final branch, though the size of the set $Y$ remains the same, the size of the set $W_l$ increases by at least 1. 
Hence, in this branch, $\mu(I^\prime)\leq \mu(I)-1$. Thus, we have shown that in each branch, the measure drops, hence completing the proof of correctness of the algorithm.\\

\noindent
{\bf Analysis.} 
Since $D$ exludes $K_h$ as a minor, Lemma~\ref{lem:minorfree_small_degree}, combined with the fact that we set $d_b = h-2$, implies that $d_w^{\max}=ch^4$, for some $c$, is an upper bound on the maximum $d_w$ which can appear during the execution of the algorithm. 
We first bound the number of leaves of the recursion tree as follows. The number of leaves is bounded by $\sum_{i=0}^{d_bk} \binom{d_bk}{i} (d_w^{\max})^{k-{\frac i {d_b}}}$. 
To see this, observe that each branch of the recursion tree can be described by a length-$d_bk$ vector as shown in the correctness paragraph. We then select $i$ positions of this vector on which the last branch was taken. Finally for  $k-{\frac i {d_b}}$ of the remaining positions, we describe which of the first at most $d_w^{\max}$ branches was taken. 
Any of the first $d_w^{max}$ branches can be taken at most $k-{\frac i {d_b}}$ times if the last branch is taken $i$ times.

The time taken along each root to leaf path in the recursion tree is polynomial, while the time taken at a leaf  for which the last branch was taken $i$ times is $\bigoh^*(2^{d_b (k - (k-{\frac i {d_b}}))+ k-{\frac i {d_b}}})=\bigoh^*(2^{i+k})$ (see Lemmata \ref{lem:base_case} and \ref{lem:base_case_nederlof}). Hence, the running time of the algorithm is 
\[
\bigoh^* \left(\sum_{i=0}^{d_bk} \binom{d_bk}{i} (d_w^{\max})^{k-{\frac i {d_b}}}\cdot 2^{i+k} \right)
=\bigoh^*\left((2d_w^{\max})^k \cdot \sum_{i=0}^{d_bk} \binom{d_bk}{i} \cdot 2^i \right) =\bigoh^*\left((2d_w^{\max})^k \cdot 3^{d_b k} \right).
\]
For $d_b= h-2$ and $d_w^{\max} =ch^4$
 this is $\bigoh^*(3^{hk + o(hk)})$. This completes the proof of the theorem. 
\end{proof}

\begin{lemma}\label{lem:littleo}
For every function $g(n) = o(\log n)$, there is a function $f(k)$ such that for every $k$ and $n$ we have $2^{g(n)k} \leq f(k) \cdot n$. 
\end{lemma}

\begin{proof}
We know that there is a function $f'(k)$ such that for every $n > f'(k)$ we have $g(n) < (\log n)/k$.
Now let $f(k)$ be the function defined as $f(k) =\max_{1 \leq n \leq f'(k)}\{2^{g(n)k}\}$.
Then, for every $k$ if $n  \leq f'(k)$ then $2^{g(n)k} \leq \max_{1 \leq n \leq f'(k)}\{2^{g(n)k}\} = f(k)$ while for $n  > f'(k)$ we have $2^{g(n)k} \leq 2^{(\log n/k) \cdot k} = 2^{\log n} = n$. Hence, indeed, $2^{g(n)k} \leq f(k) \cdot n$ for every $n$ and $k$.

%
%
%
%
%
\end{proof}

\noindent
Theorem~\ref{thm:minorfree_algo} along with Lemma~\ref{lem:littleo} has the following corollary.

\begin{corollary}
 If $\mathcal{C}$ is a class of digraphs excluding $o(\log n)$-sized minors, then {\dst} parameterized by $k$ is {\FPT} on $\mathcal{C}$.
\end{corollary}

%% file: dsttopominor.tex
We begin by observing that on graphs excluding $K_h$ as a \emph{topological minor}, we cannot apply Rule~\ref{rul:scc} since contractions may create new topological minors. Hence, we do not have the notion of a source terminal, which was crucial in designing the algorithm for this problem on graphs excluding minors. However, we will use a decomposition theorem of Grohe and Marx (\cite{GroheM12}, Theorem 4.1) to obtain a number of subproblems where we will be able to apply all the ideas developed in the previous subsection, and finally use a dynamic programming approach over this decomposition to combine the solutions to the subproblems.

\begin{theorem} {\sc (}Global Structure Theorem, {\sc \cite{GroheM12})}\label{thm:structure theorem}  For every $h\in {\mathbb N}$, there exists constants $a(h)$, $b(h)$, $c(h)$, $d(h)$, $e(h)$, such that the following holds. Let $H$ be a graph on $h$ vertices. Then, for every graph $G$ with $H\not\preceq_T G$, there is a tree decomposition $(M,\beta)$ of adhesion at most $a(h)$ such that for all $t\in V(M)$, one of the following three conditions is satisfied: 
\begin{enumerate}
\item $\vert \beta(t)\vert \leq b(h)$.
\item $\tau(t)$ has at most $c(h)$ vertices of degree larger than $d(h)$.
\item $K_{e(h)}\not\preceq \tau(t)$.
\end{enumerate}
\noindent
Furthermore, there is an algorithm that, given graphs $G$, $H$ of sizes $n$, $h$, respectively, in time $f(h)n^{\bigoh(1)}$ for some computable function $f$, computes either such a decomposition $(M,\beta)$ or a subdivision of $H$ in $G$.
\end{theorem}

\noindent
Let $(M,\beta)$ a tree decomposition given by the above theorem. Without loss of generality we assume, that for every $t \in V(M)$ we have $r \in \beta(t)$. This might increase $a(h)$, $b(h)$, $c(h)$, and $e(h)$ by at most one. For the rest of this subsection we work with this tree decomposition.

\begin{theorem}\label{thm:fpt_top_algo}
 {\dst} can be solved in time $\bigoh^*(f(h)^k)$ on graphs excluding $K_h$ as a topological minor.
\end{theorem}

\newcommand{\tab}{\textsf{Tab}}
\newcommand{\Talpha}{T_{\alpha(t)}}
\newcommand{\Tsigma}{T_{\sigma(t)}}
\newcommand{\Tgamma}{T_{\gamma(t)}}

\begin{proof}
Our algorithm is based on dynamic programming over the tree decomposition $(M,\beta)$. For $t \in V(M)$ let $T_{\sigma(t)}= (T \cup \{r\}) \cap \sigma(t)$ and $T_{\gamma(t)}= (T \cup \{r\}) \cap \gamma(t)$. 
For every $t \in V(M)$ we have one table $\tab_t$ indexed by $(R,F)$, where $\Tsigma \subseteq R \subseteq \sigma(t)$ and $F$ is a set of arcs on $R$. The index of a table represents the way a possible solution tree can cross the cut-set $\sigma(t)$. More precisely, we look for a set $S \subseteq \alpha(t)$ such that in the digraph $D[\Tgamma \cup R \cup S] \cup F$ there is a directed path from $r$ to every $t' \in \Tgamma \cup R$. In such a case we say that $S$ is \emph{good} for $t,R,F$.

For each index $R,F$ we store in $\tab_t(R,F)$ one good set $S$ of minimum size. If no such set exists, or $|S| > k$ for any such set, we set $|\tab_t(R,F)|= \infty$ and use the dummy symbol $S_\infty$ in place of the set. Naturally, $S_\infty \cup S = S_\infty$ for any set $S$.  Furthermore, if $\tab_t(R,F)= S \neq S_\infty$ we let $\kappa_t(R,F)$ be the set of arcs on $R$ such that $(u,v) \in \kappa_t(R,F)$ iff there is a directed path from $u$ to $v$ in $D[\Tgamma \cup R \cup \tab_t(R,F)]$. Note also, that if $|\tab_t(R,F)|=0$ then $\kappa_t(R,F)$ only depends on $R$, not on $F$. As $\sigma$ of the root node of $M$ is $\emptyset$, the only entry of $\tab$ for root is an optimal Steiner tree in $D$. Let us denote by $g(h)$ the maximum number of entries of the table $Tab_t$ over $t \in V(M)$. It is easy to see that $g(h) \leq 2^{a(h)+a(h)^2}$. 

The algorithm to fill the tables proceeds bottom-up along the tree decomposition and we assume that by the time we start filling the table for $t$, the tables for all its proper descendants have already been already filled. We now describe the algorithm to fill the table for $t$, distinguishing three cases, based on the type of node $t$ (see Theorem~\ref{thm:structure theorem}).


\subsubsection{Case 1: $\tau(t)$ has at most $c(h)$ vertices of degree larger than $d(h)$.}

In this case we use Algorithm~\ref{algo:small_deg}. For each $R$ and $F$ it first removes the irrelevant parts of the graph and then branches on the non-terminal vertices of high degree. 
Following that, it invokes Algorithm~\ref{algo:sat_child}. 
Note that, since $t$ can have an unbounded number of children in $M$ we cannot afford to guess the solution for each of them. Hence \textsc{SatisfyChildrenSD} only branches on the solution which is taken from the children which need at least one private vertex of the solution. After a solution is selected for all such children, it uses Rule~\ref{rul:scc} and unless the number of obtained source terminals is too big, in which case there is no solution for the branch, it uses the modified algorithm of Nederlof as described in Lemma~\ref{lem:base_case_nederlof}.

For the proof of the correctness of the algorithm, we need several observations and lemmas.

\begin{algorithm}
\lForEach{$R$ with  $\Tsigma \subseteq R \subseteq \sigma(t)$}{\\
\ForEach{$F \subseteq R^2$}{
$D'\leftarrow D[\alpha(t) \cup R] \cup F$.\\
 $S\leftarrow S_\infty$.\\
 $B \leftarrow \{v \mid v \in (\beta(t) \cap V(D')) \setminus (T \cup R) \& \deg_{\tau(t)}(v) > d(h)\}$.\\
 \ForEach{$Y \subseteq B$ with $|Y| \leq k$}{
$D'' \leftarrow D'\setminus(B\setminus Y)$.\\
$M' \leftarrow$ subtree of $M$ rooted at $t$.\\
 Let $\beta': V(M') \to 2^{V(D'')}$ be such that $\beta'(s) = \beta(s) \cap V(D'')$ for every $s$ in $V(M')$.\\
 $S' \leftarrow$ \textsc{SatifyChildrenSD}$(D'',r,T \cup Y \cup R \setminus \{r\},k - |Y|, M', \beta') \cup Y$.\\
 \lIf{$|S'| < |S|$}{$S \leftarrow S'$}\\ 
 }

$\tab_t(R,F) \leftarrow S$.\\
}
}
 \BlankLine
  \caption{Algorithm {\sc SmallDeg} to fill $\tab_t$ if all but few vertices of the bag have small degrees.}\label{algo:small_deg}
\end{algorithm}

\newcounter{sc_else_counter}

\begin{algorithm}
\SetKwInOut{Input}{Input}\SetKwInOut{Output}{Output}
\Input{An instance $(D',r, T', k)$ of DST, a tree decomposition $(M, \beta)$ rooted at $t$}
\Output{A smallest solution to the instance or $S_\infty$ if all solutions are larger than $k$}
\lIf{$k < 0$}{\Return $S_\infty$}\\
\eIf{$\exists s$ child of $t$ such that $|\tab_s(T'_{\sigma(s)}, \{r\} \times (T' \cap \sigma(s))| >0$\label{algo_sc:condition}}{
  $S\leftarrow S_\infty$.\\ 
  \lForEach{$R'$ s.t. $T'_{\sigma(s)} \subseteq R' \subseteq \sigma(s)$}{\\
     \lForEach{$F' \subseteq (R')^2$}{\\
        \If{$|\tab_s(R',F') \cup (R' \setminus T'_{\sigma(s)})| \leq k$}{
            $\hat{D} \leftarrow D' \setminus (\gamma(s) \setminus R')$.\label{algo_sc:inst_start}\\
	    $D''\leftarrow \hat{D} \cup \kappa_s(R',F')$.\\
	    $T'' \leftarrow (T' \cap V(D'')) \cup (R' \setminus T'_{\sigma(s)})$.\\
	    $k' \leftarrow k - |\tab_s(R',F')| - |R' \setminus T'_{\sigma(s)}|$.\label{algo_sc:inst_end}\\
	    $M' \leftarrow M$ with the subtree rooted at $s$ removed.\\
             Let $\beta': V(M') \to 2^{V(D'')}$  be such that $\beta'(s) = \beta(s) \cap V(D'')$ for every $s$ in $V(M')$.\\
	    $S' \leftarrow$ \textsc{SatisfyChildrenSD} $(D'',r, T'', k', M', \beta', Y') \cup \tab_s(R',F') \cup (R' \setminus T'_{\sigma(s)})$.\\
	   \lIf{$|S'| < |S|$}{$S \leftarrow S'$}\\
         }
       }
     }
   }{\setcounter{sc_else_counter}{\value{AlgoLine}} \addtocounter{sc_else_counter}{-1}
    Apply Rule 1 exhaustively to $(D',r, T', k)$ to obtain $(D'',r,T'',k)$\label{algo_sc:else_start}.\\
    Denote by $T_0$ the source terminals in $(D'',r, T'', k)$.\label{algo_sc:sources}\\
    \lIf {$|T_0| > k\cdot \max\{d(h),a(h)\}$}{$S\leftarrow S_\infty$\\}
     \lElse{$S \leftarrow$ \textsc{Nederlof} $(D'',r, T'', T_0)$.\label{algo_sc:else_end}}
   }
\lIf{$|S| > k$}{$S \leftarrow S_\infty$}\\
\Return $S$ \\

\caption{Function \textsc{SatisfyChildrenSD} $(D',r, T', k, M, \beta)$ doing the main part of the work of Algortihm~\ref{algo:small_deg}.} \label{algo:sat_child}
\end{algorithm}

\begin{observation}\label{obs:solution_terminals}
 Let $(D,r,T,k)$ be an instance of DST. For every $x$ in $V \setminus (T \cup \{r\})$ we have that $S$ is a solution for $(D,r,T \cup \{x\}, k-1)$ if and only if $S \cup \{x\}$ is a solution for $(D,r,T,k)$ and $x$ is reachable from $r$ in $D[T\cup S \cup \{r,x\}]$. In particular, if $S \cup \{x\}$ is a minimal solution for $(D,r,T,k)$, then $S$ is a minimal solution for $(D,r,T \cup \{x\}, k-1)$.
\end{observation}


\begin{lemma}\label{lem:children_replaced}
Let $(D',r,T',k)$ be an instance of DST and $M, \beta$ be a tree decomposition for $D$ rooted at $t$. Let $s$ be a child of $t$ for which the condition on line~\ref{algo_sc:condition} of Algorithm~\ref{algo:sat_child} is satisfied. Let $(D'',r,T'',k')$ be the instance as formed by lines \ref{algo_sc:inst_start}--\ref{algo_sc:inst_end} of the algortihm on $s$ for some $R'$ and $F'$. If $S$ is a solution for $(D'',r,T'',k')$ then $S'=S \cup \tab_s(R',F') \cup (R' \setminus T'_{\sigma(s)})$ is a solution for $(D',r,T',k)$.
\end{lemma}

\begin{proof}
Obviously $|S'| \leq k$. We have to show that every vertex $t'$ in $T'$ is reachable from $r$ in $D'[T' \cup \{r\} \cup S']$. For a vertex $t'$ in $T''$ there is a path from $r$ to $t'$ in $D''$. If this path contains an arc $(u,v)$ in $\kappa(R',F')$, then there is a path from $u$ to $v$ in $D[T'_{\gamma(s)} \cup R' \cup \tab_s(R',F')]$, and we can replace the arc $(u,v)$ with this path, obtaining (possibly after shortcutting) a path in $D'[T' \cup \{r\}\cup S']$. For a vertex $t'$ in $T' \cap \gamma(s)$ there is a path $P$ from $r$ to $t'$ in $D[T'_{\gamma(s)} \cup R' \cup \tab_s(R',F')] \cup F'$ as the $\tab_s$ was filled correctly. Let $r'$ be the last vertex of $R'$ on $P$. Replacing the part of $P$ from $r$ to $r'$ by a path in $D'[T' \cup \{r\}\cup S']$ obtained in the previous step, we get (possibly after shortcutting) a path from $r$ to $t'$ in $D'[T'\cup \{r\} \cup S']$ as required.\end{proof}

\begin{observation}\label{obs:R_reach}
If $s$ is a child of $t$ and there are some $R'$ and $F'$ such that $|\tab_s(R', F')| =0$ then every vertex in $T' \cap \gamma(s)$ is reachable from some vertex in $R'$ in $D'[R'\cup T'_{\gamma(s)}]$. Furthermore, if there is an $F'$ such that $|\tab_s(R', F')| =0$, then $|\tab_s(R', \{r\} \times (R' \setminus \{r\})| =0$, and also $|\tab_s(R'', \{r\} \times (R'' \setminus \{r\})| =0$ for every $R' \subseteq R'' \subseteq \sigma(s)$.
\end{observation}

\begin{proof}
The first statement follows from the definition of $\tab_s$, as $F'$ is a set of arcs on $R'$. The second part is a direct consequence of the first.
\end{proof}

\begin{observation}\label{obs:sources_on_bound}
Suppose the condition on line~\ref{algo_sc:condition} of Algorithm~\ref{algo:sat_child} is not satisfied, $(D'',r,T'',k)$ is obtained from $(D',r,T',k)$ by exhaustive application of Rule~\ref{rul:scc} and $T_0$ is the set of source terminals. Let $t_0 \in T_0$ be obtained by contracting a strongly connected component $C$ of $D'[T']$. If there is a vertex $u \in C \cap \gamma(s)$ for some child $s$ of $t$, then there is a vertex $v \in C \cap \sigma(s)$.
\end{observation}

\begin{proof}
Since the condition is not satisfied, it follows from Observation~\ref{obs:R_reach} that there is a path from some $v \in \sigma(s)$ to $u$ in $D'[T'\cup \{r\}]$. Since $t_0$ is a source terminal, this path has to be fully contained in $C$.
\end{proof}

\begin{lemma}\label{lem:domin_bound}
Suppose the condition on line~\ref{algo_sc:condition} of Algorithm~\ref{algo:sat_child} is not satisfied, $(D'',r,T'',k)$ is obtained from $(D',r,T',k)$ by exhaustive application of Rule~\ref{rul:scc} and $T_0$ is the set of source terminals.  A vertex $x$ in $\alpha(s) \setminus (T' \cup \{r\})$ for some child $s$ of $t$ can dominate at most $a(h)$ vertices of $T_0$. A vertex $x$ in $\beta(t) \setminus (T' \cup \{r\})$ can dominate at most $deg_{\tau(t)}(x)$ vertices of $T_0$.
\end{lemma}

\begin{proof}
If $x \in V(D') \setminus T'$ dominates a vertex $t_0 \in T_0$, then $t_0$ was obtained by contracting some strongly connected component $C$ of $D'[T']$ and there is an $y \in C \cap N^+_{D'}(x)$. 
If $x$ is in $\alpha(s)$, then $N^+_{D'}(x) \subseteq \gamma(s)$, $C$ contains a vertex of $\sigma(s)$ due to Observation~\ref{obs:sources_on_bound}, and, hence, there can be at most $a(h)$ such $t_0$'s.
If $x$ is in $\beta(t)$, then either $y$ is also in $\beta(t)$, in which case the edge $xy$ is in $\tau(t)$, or $y$ is in $\alpha(s)$ for some child $s$ of $t$. In this case $x$ is in $\sigma(s)$, $C$ contains a vertex $y'$ of $\sigma(s)$ due to Observation~\ref{obs:sources_on_bound}, and we can account $t_0$ to the edge $xy'$ of $\tau(t)$.
\end{proof}

\begin{lemma}\label{lem:solution_split}
Let $(D',r,T',k)$ be an instance of DST and $(M,\beta)$ be a tree decomposition for $D'$ rooted at $t$. Let $s$ be a child of $t$ for which the condition on line~\ref{algo_sc:condition} of Algorithm~\ref{algo:sat_child} is satisfied. 
Let $S$ be a solution for $(D',r,T',k)$, $R'= (T' \cup S \cup \{r\}) \cap \sigma(s)$ and $F'$ be the set of arcs $(uv)$ on $R$ such that $v$ is reachable from $u$ in $D'[(T'\cup S \cup \{r\}) \setminus \alpha(s)]$. 
Let $(D'',r,T'',k')$ be the instance as formed by lines \ref{algo_sc:inst_start}--\ref{algo_sc:inst_end} on $s$ for $R'$ and $F'$. Then $S \setminus \alpha(s) \setminus (R' \setminus T'_{\sigma(s)})$ is a solution for $(D'',r,T'',k')$, while $(S \cap \alpha(s))$ is good for $s,R',F'$. 
\end{lemma}

\begin{proof}
We  first show that every vertex of $(T' \cup S) \setminus \alpha(s)$ is reachable from $r$ in $D''=D'[(T'\cup S \cup \{r\}) \setminus \alpha(s)] \cup \kappa(R',F')$. For every vertex $t' \in R'$ there is a path from $r$ to $t'$ in $D'[T'_{\gamma(s)} \cup R' \cup \tab_s(R',F')] \cup F'$ since $\tab_s$ is correctly filled. Replacing the parts of this path in $\alpha(s)$ by arcs of $\kappa(R',F')$ and arcs of $F'$ by paths in $D'[(T'\cup S \cup \{r\}) \setminus \alpha(s)]$ one obtains a path in $D''$ to every vertex of $R'$. Now for every vertex $t'$ in $(T' \cup S) \setminus \alpha(s)$ there is a path $P$ from $r$ to $t$ in $D'[T' \cup S \cup \{r\}]$. Let $r'$ be the last vertex of $P$ in $R'$. Then concatenating the path from $r$ to $r'$ obtained in the previous step with the part of $P$ between $r'$ and $t'$ we get a path from $r$ to $t'$ in $D''$. 

We have shown, that $S \setminus \alpha(s)$ is a solution for $(D'',r,T',k-|\tab_s(R',F')|)$. It remains to use Observation~\ref{obs:solution_terminals} to show that $S \setminus \alpha(s) \setminus (R' \setminus T'_{\sigma(s)})$ is a solution for $(D'',r,T'',k')$.The second claim follows from that there is a path from $r$ to every $t' \in (T' \cup S) \cap \gamma(s) \supseteq T'_{\gamma(s)} \cup R'\setminus \{r\}$ in $D'[T' \cup S \cup \{r\}]$ and the parts of it outside $\gamma(s)$ can be replaced by arcs of $F'$.
\end{proof}

\begin{lemma}\label{lem:sc_minimal}
 Let $(D',r,T',k)$ be an instance of DST and let $(M,\beta)$ be a tree decomposition for $D'$. 
If there is a solution $S$ of size at most $k$ for $(D',r,T',k)$, 
then the invocation  \\ \textsc{SatisfyChildrenSD}$(D',r,T',k,M,\beta)$ returns a set $S'$ not larger than $S$.
\end{lemma}

\begin{proof}
We prove the claim by induction on the depth of the recursion. Note that the depth is bounded by the number of children of $t$ in $M$. Suppose first that the condition on line~\ref{algo_sc:condition} is not satisfied and $|T_0| > k\cdot \max\{d(h),a(h)\}$. As no vertex can dominate more than $\max\{d(h),a(h)\}$ vertices of $T_0$ by Lemma~\ref{lem:domin_bound}, there is a vertex of $T_0$ not dominated by $S$, which is a contradiction. If $|T_0| \leq k\cdot \max\{d(h),a(h)\}$, it follows from the optimality of the modified Nederlof's algorithm (see Lemma~\ref{lem:nederlof}) and Observation~\ref{obs:t0} that $|S'| \leq |S|$.

Now suppose that the condition on line~\ref{algo_sc:condition} is satisfied for some $s$. Let $R'$, $F'$ and $(D'',r,T'',k')$ be as in Lemma~\ref{lem:solution_split}. Then $S \setminus \alpha(s) \setminus (R' \setminus T'_{\sigma(s)})$ is a solution for $(D'',r,T'',k')$, while $S \cap \alpha(s)$ is good for $s,R',F'$. Therefore $|S \cap \alpha(s)| \leq |\tab_s(R',F')|$ as $\tab_s$ is filled correctly by assumption. Moreover \textsc{SatisfyChildrenSD}$(D'',r,T'',k',M,\beta)$ will return a set $S'$ with $|S'| \leq |S \setminus \alpha(s) \setminus (R' \setminus T'_{\sigma(s)})|$ due to the induction hypothesis. 
Together we get that $|S' \cup \tab_s(R',F') \cup (R' \setminus T'_{\sigma(s)})| \leq |S|$ and therefore also the set returned by \textsc{SatisfyChildrenSD} is not larger than $S$.
\end{proof}

Now we are ready to prove the correctness of the algorithm. We first show that if the algorithm stores a set $S \neq S_\infty$ in $\tab_t(R,F)$, then $S \subseteq \alpha(t)$ and in the digraph $D[\Tgamma \cup R \cup S] \cup F$ there is a directed path from $r$ to every $t' \in \Tgamma \cup R$. This will follow from Observation~\ref{obs:solution_terminals} if we prove that \textsc{SatisfyChildrenSD}$(D',r,T',k,M,\beta)$ returning a set $S \neq S_\infty$ implies that $S$ is a solution for $(D',r,T',k)$.
We prove this claim by induction on the depth of the recursion.  If the condition on line~\ref{algo_sc:condition} is not satisfied (and hence there is no recursion) the claim follows from the correctness of the modified version of Nederlof's algorithm (see Lemma~\ref{lem:nederlof}) and Observation~\ref{obs:t0}. If the condition is satisfied, then the claim follows from Lemma~\ref{lem:children_replaced} and the induction hypothesis.

In order to prove that the set stored is minimal, assume that there is a set $S \subseteq \alpha(t)$ of size at most $k$ which is good for $t,R,F$. 
Let $B = \{v \mid v \in (\beta(t) \cap V(D')) \setminus (T \cup R) \& \deg_{\tau(t)}(v) > d(h)\}$, $Y= S\cap B$ and $D'' = D'\setminus(B\setminus Y)$. Without loss of generality we can assume that $S$ is minimal and, therefore, $S\setminus Y$ is a solution for  $(D'',r,T \cup Y \cup R \setminus \{r\},k - |Y|)$ by Observation~\ref{obs:solution_terminals}. Hence \textsc{SatifyChildrenSD}$(D'',r,T \cup Y \cup R \setminus \{r\},k - |Y|, M', \beta') \cup Y$ returns a set $S'$ not larger than $S\setminus Y$ due to Lemma~\ref{lem:sc_minimal} and the set stored in $\tab_t(R,F)$ is not larger than $|S' \cup Y|= |S|$ finishing the proof of correctness.

As for the time complexity, observe first, that the bottleneck of the running time of Algorithm~\ref{algo:small_deg} is the at most $2^{c(h)}$ calls of Algorithm~\ref{algo:sat_child}. Therefore, we focus our attention on the running time of Algorithm~\ref{algo:sat_child}.
Note that in each recursive call of \textsc{SatisfyChildrenSD}, by Observation~\ref{obs:R_reach}, as the condition on line~\ref{algo_sc:condition} is satisfied, either $|\tab_s(R',F')| > 0$ or $|R' \setminus T'_{\sigma(s)}|>0$ and thus, $k' < k$. There are at most $g(h)$ recursive calls for one call of the function. The time spent by \textsc{SatisfyChildrenSD} on instance with parameter $k$ is at most the maximum of $g(h)$ times the time spent on instances with parameter $k-1$ and the time spend by the modified algorithm of Nederlof on an instance with at most $k\cdot \max\{d(h),a(h)\}$ source terminals. As the time spent for $k<0$ is constant, we conclude that the running time in Case 1 can be bounded by $\bigoh^*((\max\{g(h), 2^{\max\{d(h),a(h)\}}\})^k)$. 

\subsubsection{Case 2: $K_{e(h)}\not\preceq \tau(t)$.}

The overall strategy in this case is similar to that in the previous case. Basically all the work is done by Algorithm~\ref{algo:sat_childMF}( \textsc{SatisfyChildrenMF}()), which is a slight modification of the function \textsc{SatisfyChildrenSD}. The modification is limited to the else branch of the condition on line~\ref{algo_sc:condition}, that is, to lines \ref{algo_sc:else_start}--\ref{algo_sc:else_end}, where Algorithm~\ref{algo:minor_algo} (developed in Section~\ref{sec:minor_free}) is used instead of the modified version of Nederlof's algorithm. For every $R$ and $F$ we now simply store in $|\tab_t(R,F)|$ the result of \textsc{SatisfyChildrenMF}$(D',r,T \cup R \setminus \{r\}, k, M', \beta')$, where $D' = D[\alpha(t) \cup R] \cup F$,  $M'$ is the subtree of $M$ rooted at $t$ and $\beta': V(M') \to 2^{V(D')}$ is such that $\beta'(s) = \beta(s) \cap V(D')$ for every $s$ in $V(M')$.
 
\begin{algorithm}
\setcounter {AlgoLine} {\value{sc_else_counter}}  
 \Else{
    Apply Rule 1 exhaustively to $(D',r, T', k)$ to obtain $(D'',r,T'',k)$.\\
    Denote by $T_0$ the source terminals in $(D'',r, T'', k)$.\\
    $B_h \leftarrow$ non-terminals with degree at least $\max\{e(h)-1,a(h)+1\}$ in $T_0$.\label{algo_scm:instance_start}\\
$B_l \leftarrow$ non-terminals with degree at most $\max\{e(h)-2,a(h)\}$ in $T_0$.\\
$W_h \leftarrow$ source terminals with an in-neighbor in $B_h$.\label{algo_scm:instance_end}\\
$S\leftarrow$ {\sc DST-solve}($(D'',r,T'',k),\max\{e(h)-2,a(h)\},B_h, B_l, \emptyset, W_h, T_0 \setminus W_h$).\\
    } 
\Return $S$ \\

\caption{Part of the function \textsc{SatisfyChildrenMF} $(D',r, T', k, M, \beta)$ which differs from the appropriate part of the function \textsc{SatisfyChildrenSD}.} \label{algo:sat_childMF}
\end{algorithm}

For the analysis, we need most of the lemmata proved for Case 1. To prove a running time upper bound we also need the following lemma.

\begin{lemma}\label{lem:tau_minor}
Suppose the condition on line~\ref{algo_sc:condition} of Algorithm~\ref{algo:sat_childMF} is not satisfied, $(D'',r,T'',k)$ is obtained from $(D',r,T',k)$ by exhaustive application of Rule~\ref{rul:scc}, $T_0$ is the set of source terminals and $B_h$ and $W_h$ are defined as on lines~\ref{algo_scm:instance_start}--\ref{algo_scm:instance_end}. Then $D''[B_h \cup W_h]\preceq \tau(t)$.
\end{lemma}

\begin{proof}
As proven in Lemma~\ref{lem:domin_bound}, the non-terminals in $\gamma(t) \setminus \beta(t)$ can only dominate at most $a(h)$ vertices in $T_0$. Therefore $B_h \subseteq \beta(t)$. By Lemma~\ref{obs:sources_on_bound}, each vertex $t'$ in $T_0$ was obtained by contracting a strongly connected component $C(t')$ which contains at least one vertex of $\beta(t)$. Finally, if $b \in B_h$ dominates $w \in W_h$, but there is no edge between $b$ and $\beta(t) \cap C(w)$ in $D'$, then there is a child $s$ of $t$ such that $b \in \sigma(s)$, $C(w) \cap \alpha(s) \neq \emptyset$, thus there is $y \in C(w) \cap \sigma(s)$ and $by$ is an edge of $\tau(t)$ as $\sigma(s)$ is a clique in $\tau(t)$. By the same argument $C(w) \cap \tau(t)$ is connected for every $w \in W_h$ and $D''[B_h \cup W_h]$ has a model in $\tau(t)$.
\end{proof}

%

The proof of correctness of the algorithm is similar to that in Case 1. 
By Observation~\ref{obs:solution_terminals}, it is enough to prove that if \textsc{SatisfyChildrenMF}$(D',r,T',k,M,\beta)$ returns a set $S \neq S_\infty$ then $S$ is a smallest solution for $(D',r,T',k)$.
We prove this claim again by induction on the depth of the recursion. If the condition on line~\ref{algo_sc:condition} is not satisfied (and hence there is no recursion) the claim follows from the correctness of the algorithm {\sc DST-solve} proved in Section~\ref{sec:minor_free}. If the condition is satisfied, it follows from Lemma~\ref{lem:children_replaced} and the induction hypothesis that the set returned is indeed a solution. The minimality in this case is proved exactly the same way as in Lemma~\ref{lem:sc_minimal}.


As for the time complexity, let us first find a bound $d_w^{\max}$ for {\sc DST-solve} in the case the condition on line~\ref{algo_sc:condition} is not satisfied. Lemma~\ref{lem:tau_minor} implies that $K_{e(h)}\not\preceq D''[B_h \cup W_h]$ in this case. Using Lemma~\ref{lem:minorfree_small_degree}, we can derive an upper bound $d_w^{\max} \leq c \cdot e(h)^4$ for some constant $c$. It follows then from the proof in Section~\ref{sec:minor_free} that {\sc DST-solve} runs in $\bigoh^*((2^{\bigoh(d_b)}\cdot d_w^{\max})^k)$ time and, as $d_b=\max\{e(h)-2,a(h)\}$, there is a constant $g'(h)$, such that the running time of {\sc DST-solve} can be bounded by $\bigoh^*((g'(h))^k)$. From this, similarly as in Case 1, it is easy to conclude, that the running time of the overall algorithm for Case 2 can be bounded by $\bigoh^*((\max\{g(h), g'(h)\})^k)$.

\subsubsection{Case 3: $\vert \beta(t)\vert \leq b(h)$.}
If $\vert \beta(t)\vert \leq b(h)$, then no vertex in $\tau(t)$ has degree larger than $b(h)-1$, and $K_{b(h)+1}\not\preceq \tau(t)$. Therefore, in this case, either of the two approaches described above can be used. 
This completes the proof of Theorem~\ref{thm:fpt_top_algo}.
\end{proof}

%% file: dstdegenerate.tex
\begin{figure}[t]
 \centering
 \includegraphics[width=250 pt,height=200 pt]{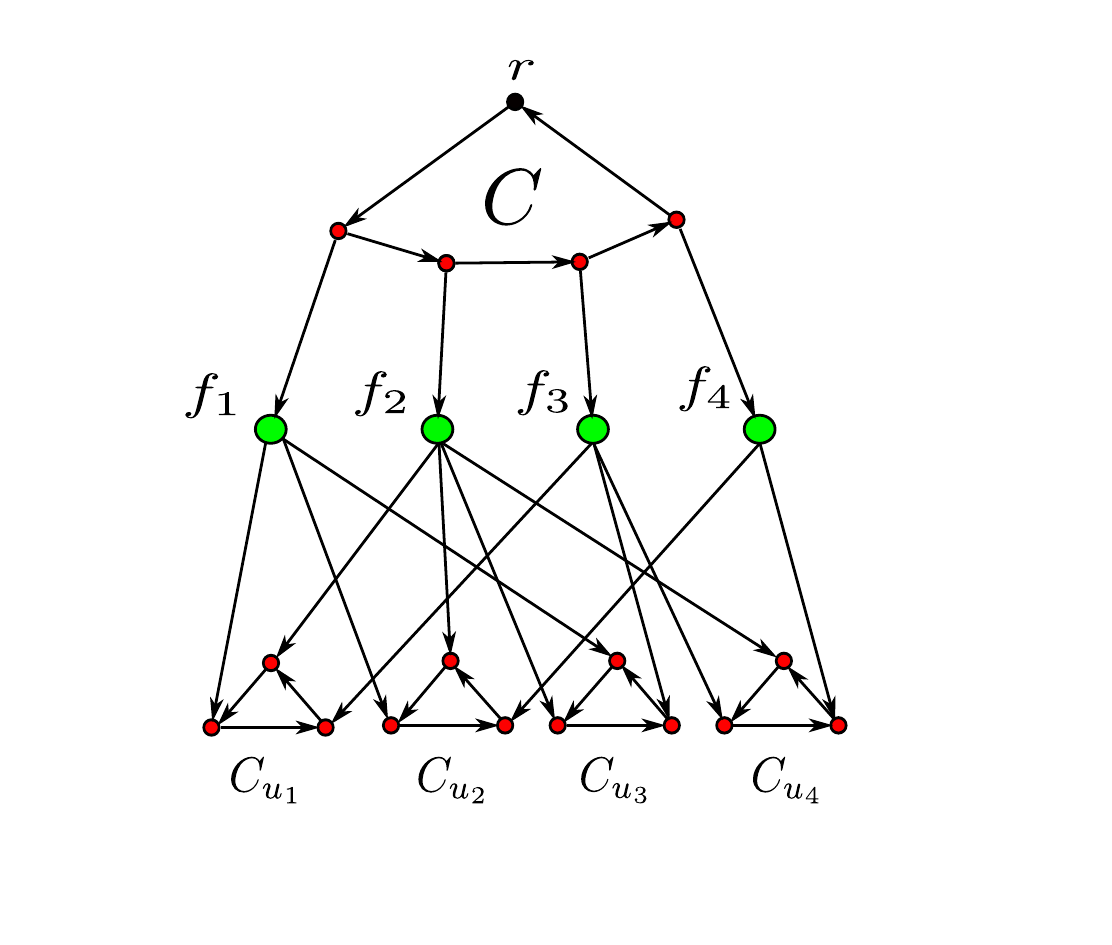}
\caption{A instance of {\sc Set Cover} reduced to an instance of {\dst}. The red vertices are the terminals and the green vertices are the non-terminals.}
\label{fig:set cover}
\end{figure}
Since {\dst} has a $O^*(f(k,h))$ algorithm on graphs excluding minors and topological minors, a natural question is if {\dst} has a $O^*(f(k,d))$ algorithm on $d$-degenerated graphs. However, we show that in general, we cannot expect an algorithm of this form even for an arbitrary 2-degenerated graph.

\begin{theorem}\label{thm:2deg_hardness}
  {\dst} parameterized by $k$ is W[2]-hard on 2-degenerated graphs. 
\end{theorem}

\begin{proof} The proof is by a parameterized reduction from {\sc Set Cover}. Given an instance $({\cal U},{\cal F}=\{F_1,\dots,F_m\},k)$ of {\sc Set Cover}, we construct an instance of {\dst} as follows. Corresponding to each set $F_i$, we have a vertex $f_i$ and corresponding to each element $u\in {\cal U}$, we add a directed cycle $C_u$ of length $l_u$ where $l_u$ is the number of sets in $\cal F$ which contain $u$ (see Fig.~\ref{fig:set cover}). For each cycle $C_u$, we add an arc from each of the sets containing $u$, to a unique vertex of $C_u$. Since $C_u$ has $l_u$ vertices, this is possible. Finally, we add another directed cycle $C$ of length $m+1$ and for each vertex $f_i$, we add an arc from a unique vertex of $C$ to $f_i$. Again, since $C$ has length $m+1$, this is possible. Finally, we set as the root $r$, the only remaining vertex of $C$ which does not have an arc to some $f_i$ and we set as terminals all the vertices involved in a directed cycle $C_u$ for some $u$ and all the vertices in the cycle $C$ except the root $r$. It is easy to see that the resulting digraph has degeneracy $3$. Finally, we subdivide every edge which lies in a cycle $C_u$ for some $u$, or on the cycle $C$ and add the new vertices to the terminal set. This results in a digraph $D$ of degeneracy 2.  Let $T$ be the set of terminals as defined above. This completes the construction. We claim that $({\cal U},{\cal F},k)$ is a {\Yes} instance of {\sc Set Cover} iff $(D,r,T,k)$ is a {\Yes} instance of {\dst}.

Suppose that $({\cal U},{\cal F},k)$ is a {\Yes} instance and let $F\subseteq {\cal F}$ be a solution. Consider the set $F_v=\{f_i\vert F_i\in F\}$. Clearly, $\vert F\vert \leq k$ and $F$ is a solution for the instance $(D,r,T,k)$ as all the terminals are reachable from $r$ in $D[F \cup T \cup \{r\}]$.

Conversely, suppose that $F_v$ is a solution for $(D,r,T,k)$. Since the only non-terminals are the vertices corresponding to the sets in $\cal F$, we define a set $F\subseteq {\cal F}$ as $F=\{F_i\vert f_i\in F_v\}$. Clearly $\vert F\vert\leq k$. We claim that $F$ is a solution for the {\sc Set Cover} instance $({\cal U},{\cal F},k)$. Since there are no edges between the cycles $C_u$ or $C$ in the instance of {\dst}, for every $u$, it must be the case that $F_v$ contains some vertex $f_i$ which has an arc to a vertex in the cycle $C_u$. But the corresponding set $F_i$ will cover the element $u$ and we have defined $F$ such that $F_i\in F$. Hence, $F$ is indeed a solution for the instance $({\cal U},{\cal F},k)$. This completes the proof.\end{proof}

\noindent
In the instance of {\dst} obtained in the above reduction, it seems that the presence of directed cycles in the subgraph induced by the terminals plays a major role in the {\em hardness} of this instance. We formally show that this is indeed the case by presenting an {\FPT} algorithm for {\dst} for the case the digraph induced by the terminals is acyclic.

\begin{theorem}\label{thm:d_deg_algo}
 {\dst} can be solved in time $\bigoh^*(2^{\bigoh(dk)})$ on $d$-degenerated graphs if the digraph induced by the terminals is acyclic.
\end{theorem}

\begin{proof}
As the digraph induced by terminals is acyclic, Rule~\ref{rul:scc} does not apply and the instance is reduced. Therefore we can directly execute the algorithm \textsc{DST-solve} on it. We set the degree bound to $d_b=d$. Note that if the set $B_h$ and $W_h$ created by the algorithm fulfill the invariants, then, as the digraph induced by $W_h \cup B_h$ is $d$-degenerated and the degree of every vertex in $B_h$ is at least $d_b+1 = d+1$, there must be a vertex $v \in W_h$ with at most $d$ (in-)neighbors in $B_h$. Therefore we have $d_w^{\max} =d$ and according to the analysis from Section~\ref{sec:minor_free}, the algorithm runs in $\bigoh^*((2^{\bigoh(d_b)}\cdot d_w^{\max})^k)=\bigoh^*(2^{\bigoh(d)k})$ time.
\end{proof}

\noindent
Theorem~\ref{thm:d_deg_algo} combined with Lemma~\ref{lem:littleo} results in the following corollary.

\begin{corollary}
If $\mathcal{C}$ is an $o(\log n)$-degenerated class of digraphs, then {\dst} parameterized by $k$ is {\FPT} on $\mathcal{C}$ if the digraph induced by terminals is acyclic.
\end{corollary}

%% file: domsethardness.tex
\noindent
In this section, we show that the algorithm given in Theorem~\ref{thm:d_deg_algo} is essentially the best possible with respect to the dependency on the degeneracy of the graph and the solution size. 
We begin by proving a lower bound on the time required by any algorithm for {\dst} on graphs of degeneracy $\bigoh (\log n)$.



\noindent Our starting point is the known result for the following problem.
\begin{center}
\begin{boxedminipage}{.8\textwidth}
\decnamedefn{{\sc Partitioned Subgraph Isomorphism (PSI)}}{Undirected graphs $H=(V_H,E_H)$ and $G=(V_G=\{g_1,\dots,g_l\},E_G)$ and a coloring function $col:V_H\rightarrow [l]$.}
{ Is there an injection $\phi:V_G\rightarrow V_H$ such that for every $i\in [l]$, $col(\phi(g_i))=i$ and for every $(g_i,g_j)\in E_G$, $(\phi(g_i),\phi(g_j))\in E_H$? 
}
\end{boxedminipage}
\end{center}

\noindent
We need the following lemma by Marx~\cite{Marx07}.

\begin{lemma}{\sc (}Corollary 6.3, {\sc \cite{Marx07})}\label{lem:cis hardness}
{\sc Partitioned Subgraph Isomorphism} cannot be solved in time $f(k)n^{o(\frac{k} {\log k})}$ where $f$ is an arbitrary function and $k=\vert E_G\vert$ is the number of edges in the smaller graph $G$ unless {\ETH} fails.
\end{lemma}

\noindent
Using the above lemma, we will first prove a similar kind of hardness for a restricted version of {\sc Set Cover} (Lemma~\ref{lem:set cover hardness}). Following that, we will reduce this problem to an instance of {\dst} to prove the hardness of the problem on graphs of degeneracy $\bigoh (\log n)$.

\begin{lemma}\label{lem:set cover hardness}
There is a constant $\gamma$ such that {\sc Set Cover} with size of each set bounded by $\gamma \log m$ cannot be solved in time $f(k)m^{o(\frac{k} {\log k})}$, unless {\ETH} fails, where $k$ is the size of the solution and $m$ is the size of the family of sets.
\end{lemma}

\begin{proof}
 Let $(H=(V_H,E_H),G=(V_G,E_G),col)$ be an instance of {\sc Partitioned Subgraph Isomorphism} where $\vert V_G\vert=l$ and the function $col:V_H\rightarrow [l]$ is a coloring (not necessarily proper) of the vertices of $H$ with colors from $[l]$. We call the set of vertices of $H$ which have the same color, a \emph{color class}. We assume without loss of generality that there are no isolated vertices in $G$. Let $n$ be the number of vertices of $H$. 
For each vertex of color $i$ in $H$, we assign a $\log n $-sized subset of $2 \log n$. Since ${\binom{2 \log n}{\log n} }\geq n$, this is possible. Let this assignment be represented by the function $id:V_H\rightarrow 2^{[2 \log n]}$.

Recall that the vertices of $G$ are numbered $g_1,\dots, g_l$ and we are looking for a colorful subgraph of $H$ isomorphic to $G$ such that the vertex from color class $i$ is mapped to the vertex~$g_i$. 

\begin{figure}[t]
 \centering
 \includegraphics[width=350 pt,height=175 pt]{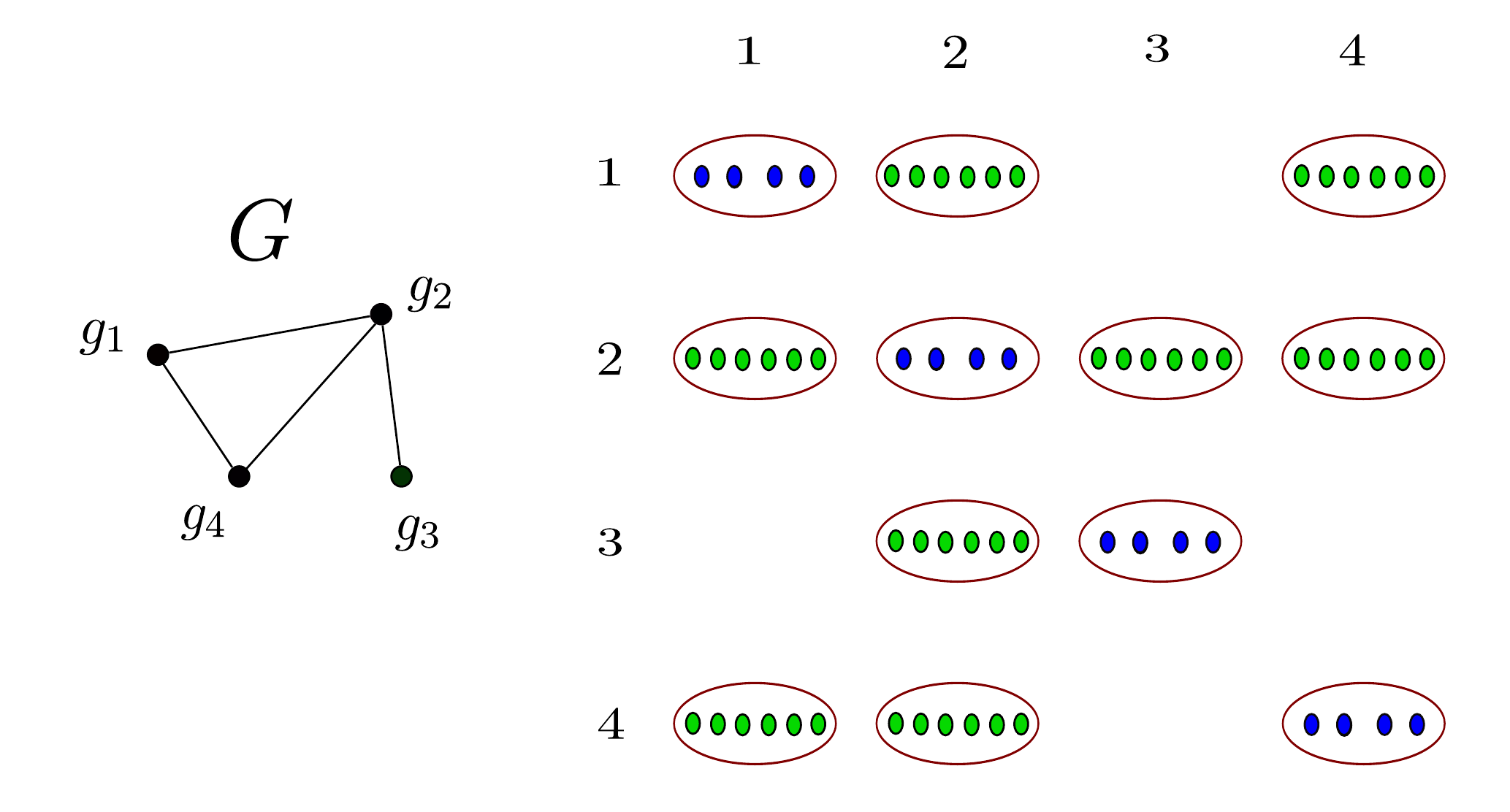}
\caption{An illustration of the sets in the reduced instance, corresponding to the graph $G$ (on the left). The blue sets at position $(i,i)$ correspond to vertices of $H$ with color $i$ and the green sets at position $(i,j)$ correspond to edges of $H$ between color classes $i$ and $j$.}
\label{fig:grid}
\end{figure}
We will list the sets of the {\sc Set Cover} instance and then we will define the set of elements contained in each set.
 For each pair $(i,j)$ such that there is an edge between $g_i$ and $g_j$, and for every edge between vertices $u$ and $v$ in $V_H$ such that $col(u)=i$, $col(v)=j$, we have a set $F^{ij}_{uv}$. For each $i\in [l]$, for each $v\in V_H$ such that $col(v)=i$, we have a set $F^{ii}_{vv}$. The notation is chosen is such way that we can think of the sets as placed on a $l \times l$ grid, where the sets $F^{ij}_{uv}$ for a fixed $i$ and $j$ are placed at the position $(i,j)$ (see Fig.~\ref{fig:grid}). Observe that many sets can be placed at a position and it may also be the case that some positions of the grid do not have a set placed on them. Let ${\cal F}$ be the family of sets defined as above.



A position $(i,j)$ which has a set placed on it is called non-empty and empty otherwise. Without loss of generality, we assume that if there are $i\neq j$ such that there is an edge between $g_i$ and $g_j$ in $G$, then the position $(i,j)$ is non-empty.
Two non-empty positions $(i,j)$ and $(i^\prime,j)$ are said to be \emph{consecutive} if there is no non-empty position $(i^{\prime\prime},j)$ where $i<i^{\prime\prime}<i^\prime$. Similarly, two non-empty positions $(i,j)$ and $(i,j^\prime)$ are said to be consecutive if there is no non-empty position $(i,j^{\prime\prime})$ where $j<j^{\prime\prime}<j^\prime$. Note that consecutive positions are only defined along the same row or column.

We now define the universe $\cal U$ as follows. For every non-empty position $(i,j)$, we have an element $s^{(i,j)}$.
For every $(i_1,j_1)$ and $(i_2,j_2)$ such that they are consecutive, we have a set ${\cal U}^{(i_1,j_1)(i_2,j_2)}$ of $2 \log n$ elements
$\{u^{(i_1,j_1)(i_2,j_2)}_1\dots,u^{(i_1,j_1)(i_2,j_2)}_{2 \log n} \}$. An element $u^{(i_1,j_1)(i_2,j_2)}_a$ is said to \emph{correspond} to $id(u)$ for some vertex $u$ if $a\in id(u)$.


%

We will now define the elements contained within each set. For each non-empty position $(i,j)$, add the element $s^{(i,j)}$ to every $F^{ij}_{uv}$ for all (possible) $u,v$. 
Now, fix $1\leq i\leq l$.
Let $(i,j_1)$ and $(i,j_2)$ be consecutive positions where $j_1<j_2$. For each set $F^{ij_1}_{uv}$, we add the elements $\{u^{(i,j_1)(i,j_2)}_a\vert a\notin id(u)\}$ and for each set $F^{ij_2}_{uv}$, we add the the elements $\{u^{(i,j_1)(i,j_2)}_a\vert a\in id(u)\}$.

Similarly, fix $1\leq j\leq l$.
Let $(i_1,j)$ and $(i_2,j)$ be consecutive positions where that $i_1<i_2$. For each set $F^{i_1j}_{uv}$, we add the elements $\{u^{(i_1,j)(i_2,j)}_a\vert a\notin id(u)\}$ and for each set $F^{i_2j}_{uv}$, we add the the elements $\{u^{(i_1,j)(i_2,j)}_a\vert a\in id(u)\}$.
%

This completes the construction of the {\sc Set Cover} instance. We first prove the following lemma regarding the constructed instance, which we will then use to show the correctness of the reduction.

\begin{lemma}\label{lem:consecutive}
Suppose $(i,j_1)$ and $(i,j_2)$ are two consecutive positions where $j_1<j_2$ and $(i_1,j)$ and $(i_2,j)$ are two consecutive positions where $i_1<i_2$.
\begin{enumerate}
\item The elements in ${\cal U}^{(i,j_1)(i,j_2)}$ can be covered by precisely one set from $(i,j_1)$ and one set from $(i,j_2)$ iff the two sets are of the form $F^{ij_1}_{uv}$ and $F^{ij_2}_{uv^\prime}$.
\item The elements in ${\cal U}^{(i_1,j)(i_2,j)}$ can be covered by precisely one set from $(i_1,j)$ and one set from $(i_2,j)$ iff the two sets are of the form $F^{i_1j}_{uv}$ and $F^{i_2j}_{u^\prime v}$.

\end{enumerate}
\end{lemma}

\begin{proof} We prove the first statement. The proof of the second is analogous. Observe that, by the construction, the only sets which can cover elements in ${\cal U}^{(i,j_1)(i,j_2)}$ are sets from $(i,j_1)$ and $(i,j_2)$.

Suppose that the elements in ${\cal U}^{(i,j_1)(i,j_2)}$ are covered by precisely one set from $(i,j_1)$ and one from $(i,j_2)$ and the two sets are of the form $F^{ij_1}_{uv}$ and $F^{ij_2}_{u^\prime v^\prime}$ where $u\neq u^\prime$. By the construction, the set $F^{ij_1}_{uv}$ covers the elements of ${\cal U}^{(i,j_1)(i,j_2)}$ which \emph{do not} correspond to $id(u)$ and the set $F^{ij_2}_{u^\prime v^\prime}$ covers the elements of ${\cal U}^{(i,j_1)(i,j_2)}$ which correspond to $id(u^\prime)$. Since $col(u)=col(u^\prime)$ (implied by the construction), $id(u)\neq id(u^\prime)$. Since $\vert id(u)\vert =\vert id(u^\prime)\vert$, it must be the case that there is an element of $[2 \log n]$, say $x$, which is in $id(u)$ but not in $id(u^\prime)$. But then, it must be the case that the element $u^{(i,j_1)(i,j_2)}_x$ is left uncovered by both $F^{ij_1}_{uv}$ and $F^{ij_2}_{u^\prime v^\prime}$, a contradiction.

Conversely, consider two sets of the form $F^{ij_1}_{uv}$ and $F^{ij_2}_{uv^\prime}$. We claim that these two sets together cover the elements in ${\cal U}^{(i,j_1)(i,j_2)}$. But this is true since $F^{ij_1}_{uv}$ covers the elements of 
${\cal U}^{(i,j_1)(i,j_2)}$ which do not correspond to $id(u)$ and $F^{ij_2}_{uv^\prime}$ covers the elements of ${\cal U}^{(i,j_1)(i,j_2)}$ which do correspond to $id(u)$. This completes the proof of the lemma.

\end{proof}

%
\noindent
We claim the instance $(H,G,col)$ is a {\Yes} instance of {\sc PSI} iff the instance $({\cal U},{\cal F},k^\prime)$ is a {\Yes} instance of {\sc Set Cover}, where $k^\prime=2\vert E_G\vert +\vert V_G\vert$.
Suppose that $(H,G,col)$ is a {\Yes} instance, $\phi$ is its solution, and let $v_i=\phi(g_i)$. 
We claim that the sets $F^{ij}_{v_iv_j}$, where $(i,j)$ is a non-empty position, form a solution for the {\sc Set Cover} instance. Since we have picked a set from every non-empty position $(i,j)$, the elements $s^{(i,j)}$ are all covered. But since the sets we picked from any two consecutive positions match premise of Lemma~\ref{lem:consecutive}, the elements corresponding to the consecutive positions are also covered.

Conversely, suppose that the {\sc Set Cover} instance is a {\Yes} instance and let ${\cal F}^\prime$ be a solution. Since we must pick at least one set from each non-empty position (we have to cover the vertices $s^{(i,j)}$), and the number of non-empty positions equals $k'$, we must have picked exactly one set from each non-empty position. Let $v_i$ be the vertex corresponding to the set picked at position $(i,i)$. 
We define the function $\phi$ as $\phi(g_i)=v_i$. Clearly, $\phi$ is an injection with $col(\phi(g_i))=i$. 
It remains to show that for every $g_i,g_j$, if $(g_i,g_j)\in E_G$, then there is an edge between $v_i$ and $v_j$. To show this, we need to show that the set picked from position $(i,j)$ has to be exactly $F^{ij}_{v_iv_j}$. By Lemma~\ref{lem:consecutive}, the sets picked from row $i$ are of the form $F^{ij}_{v_i,v}$, for any $j$ and $v$ and the sets picked from column $j$ are of the form $F^{ij}_{v,v_j}$, for any $i$ and $v$. Hence, the set picked from position $(i,j)$ can only be $F^{ij}_{v_iv_j}$. Thus, there is an edge between $v_i$ and $v_j$ in $H$ and $\phi$ is indeed a homomorphism.
This completes the proof of equivalence of the two instances.

\noindent
Since $G$ contains no isolated vertex, we have $l=O(k)$ and, thus, $k^\prime=\Theta(k)$. Observe that the 
number of sets $m$ in the {\sc Set Cover} instance is $\vert V_H\vert + 2\vert E_H\vert$, that is, $n\leq m$ and $m=\bigoh(n^2)$.
Observe that each set contains at most $4 \log n+1$ elements, one of the form $s^{(i,j)}$ and $\log n$ for each of the at most four consecutive positions the set can be a part of. Since the number of sets $m$ is at least $n$, there is a constant $\gamma$ such that the number of elements in each set is bounded by $\gamma \log m$. Finally, since $m=\bigoh(n^2)$, an algorithm for {\sc Set Cover} of the form $f(k)m^{o(\frac{k}{\log k})}$ implies an algorithm of the form $f(k)n^{o(\frac{k}{\log k})}$ for {\sc PSI}. 
This completes the proof of the lemma. \end{proof}

\noindent
Now we are ready to prove the main theorem of this section.

\begin{theorem}\label{thm:dst hardness logn}
 {\dst} cannot be solved in time $f(k)n^{o({\frac {k} {\log k}})}$ on $c\log n$-degenerated graphs for any constant $c>0$ even if the digraph induced by terminals is acyclic, where $k$ is the solution size and $f$ is an arbitrary function, unless {\ETH} fails.
\end{theorem}

\begin{proof} The proof is by a reduction from the restricted version of {\sc Set Cover} shown to be hard in Lemma~\ref{lem:set cover hardness}.  Fix a constant $c>0$ and let $({\cal U}=\{u_1,\dots,u_n\},{\cal F}=\{F_1,\dots,F_m\},k)$ be an instance of {\sc Set Cover}, where the size of any set is at most $\gamma \log m$, for some constant $\gamma$. For each set $F_i$, we have a vertex $f_i$. For each element $u_i$, we have a vertex $x_i$. If an element $u_i$ is contained in set $F_j$, then we add an arc $(f_j,x_i)$. Further, we add another vertex $r$ and add arcs $(r,f_i)$ for every $i$. Finally, we add $m^{2\gamma/c}$ isolated vertices. This completes the construction of the digraph $D$. We set $T=\{x_1,\dots, x_n\}\cup \{r\}$ as the set of terminals and $r$ as the root. 

We claim that $({\cal U}, {\cal F}, k)$ is a {\Yes} instance of {\sc Set Cover} iff $(D,r,T,k)$ is a {\Yes} instance of {\dst}. Suppose that $\{F_1,\dots,F_k\}$ is a set cover for the given instance. It is easy to see that the vertices $\{f_1,\dots, f_k\}$ form a solution for the {\dst} instance.

Conversely, suppose that $\{f_1,\dots, f_k\}$ is a solution for the {\dst} instance. Since the only way that $r$ can reach a vertex $x_i$ is through some $f_j$, and the construction implies that $u_i\in F_j$, the sets $\{F_1,\dots, F_k\}$ form a set cover for $({\cal U}, {\cal F}, k)$. This concludes the proof of equivalence of the two instances.

We claim that the degeneracy of the graph $D$ is $c \log  n_1+1$. First, we show that the degeneracy of the graph $D$ is bounded by $\gamma \log m+1$. This follows from that each vertex $f_i$ has total degree at most $\gamma \log m +1$ and if a subgraph contains none of these vertices, then it contains no edges. Now, $n_1$ is at least $m^{2\gamma/c}$. Hence, $\log n_1  \geq (2\gamma/c)\log m$ and the degeneracy of the graph is at most $\gamma \log m +1 \leq c \cdot (2\gamma/c) \log m \leq c \log n_1$. Finally, since each vertex $f_i$ is adjacent to at most $\gamma \log m +1$ vertices, $n_1=\bigoh(m \log m + m^{2\gamma/c})$ and, thus, it is polynomial in $m$. Hence, an algorithm for {\dst} of the form $f(k)n_1^{o(\frac{k}{\log k})}$ implies an algorithm of the form $f(k)m^{o(\frac{k}{\log k})}$ for the {\sc Set Cover} instance. This concludes the proof of Theorem~\ref{thm:dst hardness logn}. 
\end{proof}

\noindent
Combining the Theorem~\ref{thm:dst hardness logn} with Lemma~\ref{lem:littleo} we get the following corollary.

\begin{corollary}\label{cor:dst dependency on d}
 There are no two functions $f$ and $g$ such that $g(d)=o(d)$ and there is an algorithm for {\dst} running in time $\bigoh^*(2^{g(d)f(k)})$ unless {\ETH} fails.
\end{corollary}

\noindent
To examine the dependency on the solution size we utilize the following theorem.

\begin{theorem}\scite{ImpagliazzoPZ01}\label{thm:bounded degree subexp hard}
 There is a constant $c$ such that {\ds} does not have an algorithm running in time $\bigoh^*(2^{o(n)})$ on graphs of maximum degree $\leq c$ unless {\ETH} fails.
\end{theorem}

\noindent
From Theorem~\ref{thm:bounded degree subexp hard}, we can infer the following corollary.

\begin{corollary}\label{cor:dependency dst on k}
  There are no two functions $f$ and $g$ such that $f(k)=o(k)$ and there is an algorithm for {\dst} running in time $\bigoh^*(2^{g(d)f(k)})$, unless {\ETH} fails.
\end{corollary}

\begin{proof}
We use the following standard reduction from \ds{} to \dst{}.
 Let $(G=(V,E),k)$ be an instance of \ds{} with the maximum degree of $G$ bounded by some constant $c$. We can assume that the number of vertices $n$ of the graph $G$ is at most $ck+k$, since otherwise, it is a trivial {\No} instance. Let $D =(V',A)$ be the digraph defined as follows. We set $V' = V \times \{1,2\} \cup \{r\}$. There is an arc in $A$ from $(u,1)$ to $(v,2)$ if either $u=v$ or there is an edge between $u$ and $v$ in $E$. Finally, there is an arc from $r$ to $(v,1)$ for every $v \in V$. We let $T = \{(v,2) \mid v \in V\}$. 

 It is easy to check that $S \subseteq V$ is a solution to the instance $(G,k)$ of \ds{} if and only if $S \times \{1\}$ is a solution to the instance $(D,r,T,k)$ of \dst{}. As the vertices $(v,2)$ have degree at most $c+1$, $D$ is $c+1$-degenerated. Since the reduction is polynomial time, preserves $k$ and $k=\Theta(n)$, an algorithm for \dst{} running in time $\bigoh^*(2^{g(d)f(k)})$ for some $f(k)=o(k)$ would solve \ds{} on graphs of maximum degree $\leq c$ in time $\bigoh^*(2^{g(c+1)f(k)})=\bigoh^*(2^{o(n)})$, and, hence, \ETH{} fails by Theorem~\ref{thm:bounded degree subexp hard}.
\end{proof}

%% file: domsetalgo.tex
In this section, we adapt the ideas used in the algorithms for {\dst} to design improved algorithms for the {\ds} problem and some variants of it in subclasses of degenerated graphs.

\subsection{Introduction for \ds{}}

On general graphs {\ds} is W[2]-com\-ple\-te~\cite{DF99}. However, there are many interesting graph classes where {\FPT}-algorithms exist for {\ds}. The project of expanding the horizon where \FPT{} algorithms exist for \ds{} has produced a number of cutting-edge techniques of parameterized algorithm design. This has made {\ds} a cornerstone problem in parameterized complexity. For an example the initial study of parameterized subexponential algorithms for {\ds}, on planar graphs \cite{AlberBFKN02,FominT06} resulted in the development of bidimensionality theory characterizing a broad range of graph problems  that admit efficient approximation schemes, subexponential time \FPT{} algorithms and efficient polynomial time pre-processing (called {\em kernelization}) on minor closed graph classes~\cite{DemaineFHT05sub,DemaineHaj05}. Alon and Gutner~\cite{AlonG09} and Philip, Raman, and Sikdar~\cite{PhilipRS09} showed that \ds{} problem is {\FPT} on graphs of bounded degeneracy and on $K_{i,j}$-free graphs,  
respectively.

Numerous papers also concerned the approximability of \ds{}. It follows from~\cite{DuhF97} that \ds{} on general graphs can approximated to within roughly $\ln(\Delta(G)+1)$, where $\Delta(G)$ is the maximum degree in the graph $G$. On the other hand, it is NP-hard to approximate \ds{} in bipartite graphs of degree at most $B$ within a factor of $(\ln B - c\ln \ln B)$, for some absolute constant $c$~\cite{ChlebikC08}. Note that a graph of degree at most $B$ excludes $K_{B+2}$ as a topological minor, and, hence, the hardness also applies to graphs excluding $K_h$ as a topological minor. While a polynomial time approximation scheme (PTAS) is known for $K_h$-minor-free graphs~\cite{Grohe03}, we are not aware of any constant factor approximation for \ds{} on graphs excluding $K_h$ as a topological minor or $d$-degenerated graphs.


Based on the ideas from previous sections, we develop an algorithm for \ds{}.
Our algorithm for {\ds} on $d$-degenerated graphs improves over the $O^*(k^{O(dk)})$ time algorithm by Alon and Gutner~\cite{AlonG09}. In fact, it turns out that our algorithm is essentially optimal -- we show that, assuming the \ETH{}, the running time dependence of our algorithm on the degeneracy of the input graph and solution size $k$ cannot be significantly improved. Furthermore, we also give a factor $O(d^2)$ approximation algorithm for {\ds} on $d$-degenerated graphs. A list of our results for {\ds} is given below.

\begin{enumerate}
\item There is a $\bigoh^*(3^{hk+o(hk)})$-time algorithm for \ds{} on graphs excluding $K_h$ as a topological minor.
\item There is a $\bigoh^*(3^{dk+o(dk)})$-time algorithm for \ds{} on $d$-degenerated graphs. This implies that \ds{} is \FPT{} on $o(\log n)$-degenerated classes of graphs.
\item For any constant $c>0$, there is no $f(k)n^{o({\frac k {\log k} })}$-time algorithm on graphs of degeneracy $c \log n$ unless {\ETH} fails.
\item  There are no two functions $f$ and $g$ such that $g(d)=o(d)$ and there is an algorithm for {\ds} running in time $\bigoh^*(2^{g(d)f(k)})$, unless {\ETH} fails.
\item  There are no two functions $f$ and $g$ such that $f(k)=o(k)$ and there is an algorithm for {\ds} running in time $\bigoh^*(2^{g(d)f(k)})$, unless {\ETH} fails.
\item There is a $O(dn\log n)$ time factor $O(d^2)$ approximation algorithm for {\ds} on $d$-degenerated graphs.
\end{enumerate}

\subsection{{\ds} on graphs of bounded degeneracy}

We begin by giving an algorithm for {\ds} running in time $\bigoh^*(3^{hk+o(hk)})$ in graphs excluding $K_h$ as a topological minor. This improves over the $\bigoh^*(2^{\bigoh(kh \log h)})$ algorithm of \cite{AlonG09}. Though the algorithm we give here is mainly built on the ideas developed for the algorithm for {\dst}, the algorithm has to be modified slightly in certain places in order to fit this problem. We also stress this an example of how these ideas, with some modifications, can be made to fit problems other than {\dst}. We begin by proving lemmata required for the correctness of the base cases of our algorithm.

\begin{lemma}\label{lem:base_case_domset}
 Let $(G,k)$ be an instance of {\ds}. Let $Y\subseteq V$ be a set of vertices and let $B$ be the set of vertices (other than $Y$) dominated by $Y$. Let $W$ be the set of vertices of $G$ not dominated by $Y$, $B_h$ be the set of vertices of $B$ which dominate at least $d+1$ terminals in $W$ for some constant $d$, $B_l$ be the rest of the vertices of $B$, $W_h$ be the vertices of $W$ which have neighbors in $B_h\cup W$. If $\vert W\setminus W_h\vert >d(k-\vert Y\vert)$, then the given instance does not admit a solution which contains~$Y$.
\end{lemma}
\begin{proof}
Let $W\setminus W_h=W_l$. Since any vertex in $W_l$ does not have neighbors in $B_h\cup W\cup Y$, $G[W_l]$ is an independent set, and the only vertices which can dominate a vertex in $W_l$ are either itself, or a vertex of $B_l$. Any vertex of $W_l$ can only dominate itself (vertices of $B_l$ are already dominated by $Y$) and any vertex of $B_l$ can dominate at most $d$ vertices of $W$. Hence, if $\vert W_l\vert>d(k-\vert Y\vert)$, $W_l$ cannot be dominated by adding $k-\vert Y\vert$ vertices to $Y$. This completes the proof.  
\end{proof}

%


\begin{lemma}\label{lem:base_case_nederlof_domset}
 Let $(D,k)$ be an instance of {\ds}. 
Let $Y\subseteq V$ be a set of vertices and let $B$ be the set of vertices (other than $Y$) dominated by $Y$. Let $W$ be the set of vertices of $G$ not dominated by $Y$, $B_h$ be the set of vertices of $B$ which dominate at least $d+1$ terminals in $W$ for some constant $d$, $B_l$ be the rest of the vertices of $B$, $W_h$ be the vertices of $W$ which have neighbors in~$B_h\cup W$.

If $B_h\cup W_h$ is empty, then there is an algorithm which can test if this instance has a solution containing $Y$ in time~$\bigoh^*(2^{d(k-|Y|)})$.
\end{lemma}

\begin{proof}
The premises of the lemma imply that the only potentially non-empty sets are $Y$, $B_l$ and $W=W_l$. If $W$ is empty, we are done. Suppose $W$ is non-empty. 
 We know that $\vert Y\vert \leq k$ and, by Lemma~\ref{lem:base_case_domset}, we can assume that $\vert W\vert \leq d(k-\vert Y\vert)$. 
Since $W$ is an independent set, the only vertices the vertices of $W$ can dominate (except for vertices dominated by $Y$), are themselves. Thus it is never worse to take a neighbor of them in the solution if they have one. The isolated vertices from $W$ can be simply moved to $Y$, as we have to take them into the solution. Now, it remains to find a set of vertices of $B_l$ of the appropriate size such that they dominate the remaining vertices of $W$. But in this case, we can reduce it to an instance of {\dst} and apply the algorithm from Lemma~\ref{lem:nederlof}.
 The reduction is as follows. Consider the graph $G[B_l\cup W]$. Remove all edges between vertices in $B_l$. Let this graph be $G_{dst}$. Simply add a new vertex $r$ and add directed edges from $r$ to every vertex in $B_l$. Also orient all edges between $B_l$ and $W$, from $B_l$ to $W$. Now, $W$ is set as the terminal set. Now, since $\vert W\vert\leq dk$, it is easy to see that the algorithm \textsc{Nederlof}$(G_{dst},r, W,W)$ runs in time $\bigoh^*(2^{d(k-|Y|)})$.
 This completes the proof of the lemma. 
\end{proof}


\begin{theorem}\label{thm:top_algo_ds}
 {\ds} can be solved in time $\bigoh^*(3^{hk+o(hk)})$ on graphs excluding $K_h$ as a topological minor.
\end{theorem}

\begin{proof}

The algorithm we describe takes as input an instance $(G,k)$, and vertex sets $B,W, B_h$, $B_l$, $W_h$, $W_l$ and $Y$ and returns a smallest solution for the instance, which contains $Y$ if such a solution exists. If there is no such solution, then the algorithm returns a dummy symbol $S_\infty$. To simplify the description, we assume that $|S_\infty| = \infty $. The algorithm is a recursive algorithm and at any stage of recursion, the corresponding recursive step returns the smallest set found in the recursions initiated in this step.
For the purpose of this proof, we will always have $d_b=h-2$.  Initially, all the sets mentioned above are empty.


\begin{figure}[t]
\centering
\includegraphics[width=250 pt,height=155 pt]{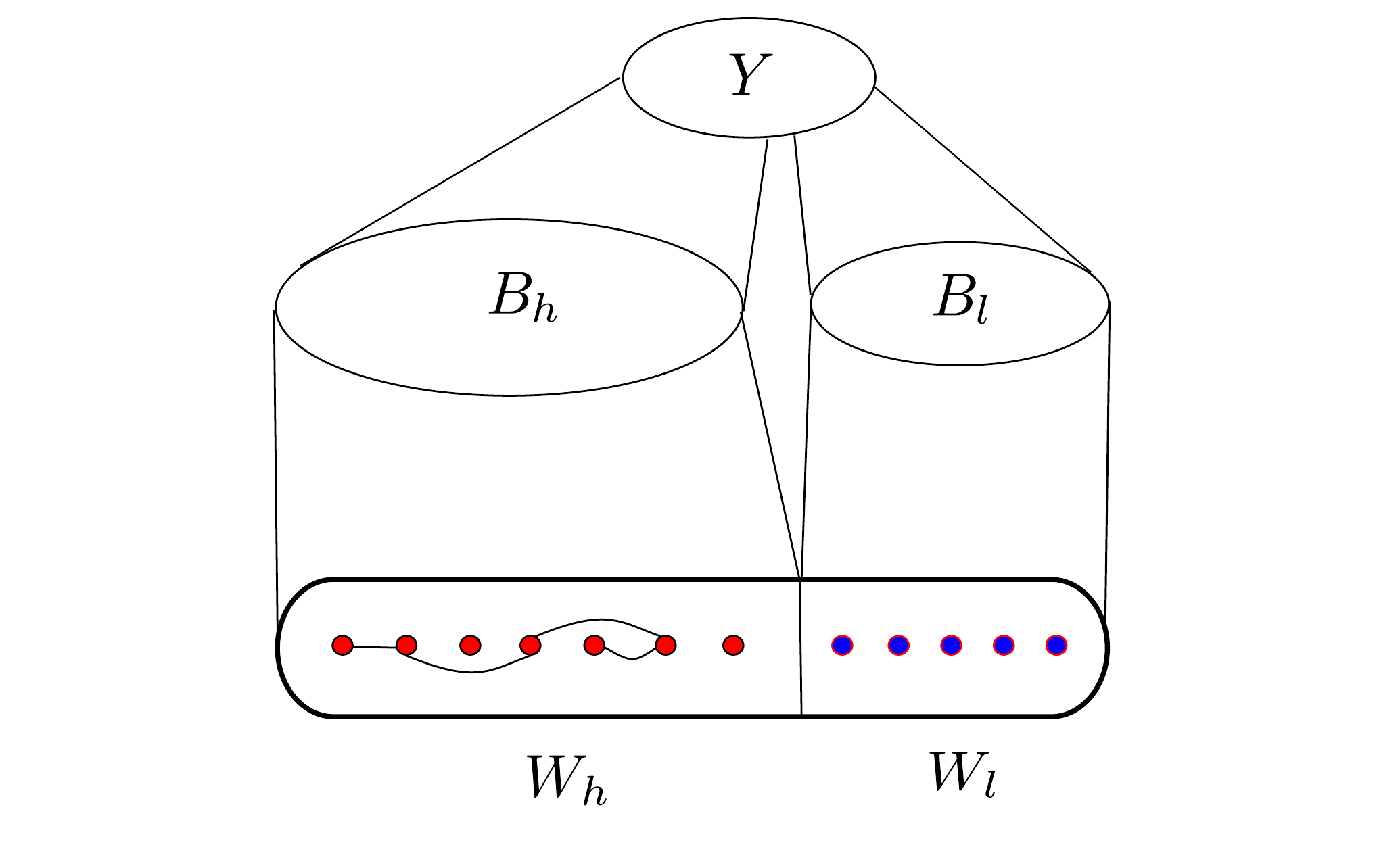}
\caption{An illustration of the sets defined in Theorem~\ref{thm:top_algo_ds}}
\label{fig:partition_domset}
\end{figure}


\noindent
At any point, while updating these sets, we will maintain the following invariants (see Fig.~\ref{fig:partition_domset}).
\begin{itemize}
\item The sets $B_h$, $B_l$, $W_h$, $W_l$ and $Y$ are pairwise disjoint.
\item The set $Y$ has size at most $k$.
\item $B$ is the set of vertices dominated by $Y$.
\item $W$ is the set of vertices not dominated by $Y$.
\item $B_h$ is the set of vertices of $B$ which dominate at least $d_b+1$ vertices of $W$.
\item $B_l$ is the set of vertices of $B$ which dominate at most $d_b$ vertices of $W$.
\item $W_h$ is the set of vertices of $W$ which have a neighbor in $B_h\cup W$.
\item $W_l$ are the remaining vertices of $W$.
\end{itemize}

\noindent

Observe that the sets $B_l$, $W_h$ and $W_l$ correspond to the sets $B_l$, $W_h$ and $W\setminus W_h$ defined in the statement of Lemma~\ref{lem:base_case_domset}. 
Hence, by Lemma~\ref{lem:base_case_domset}, if $\vert W_l\vert>d_b(k-\vert Y\vert)$, then there is no solution containing $Y$ and hence we return $S_\infty$ (see Algorithm~\ref{algo:minor_algo_domset}). 
If $B_h\cup W_h$ is empty, then we apply Lemma~\ref{lem:base_case_nederlof_domset} to solve the problem in time $\bigoh^*(2^{d_b (k-|Y|)})$. If $B_h\cup W_h$ is non-empty, then we find a vertex $v\in W_h$ with the least neighbors in $B_h\cup W_h$. Let $N$ be the set of these neighbors and let $\vert N\vert=d_w$. 

We then branch into $d_w+2$ branches described as follows. In the first $d_w+1$ branches, we move a vertex $u$ of $N \cup \{v\}$, to the set $Y$, and perform the following updates. We move the vertices of $W$ which are adjacent to $u$, to $B$, and update the remaining sets in a way that maintains the invariants mentioned above. More precisely, set $B_h$ as the set of vertices of $B$ which dominate at least $d_b+1$ vertices of $W$, $B_l$ as the set of vertices of $B$ which dominate at most $d_b$ vertices of $W$, $W_h$ as the set of vertices of $W$ which have a neighbor in $B_h\cup W$, and $W_l$ as the rest of the vertices of $W$. Finally, we recurse on the resulting instance.


In the last of the $d_w+2$ branches, that is, the branch where we have guessed that none of the vertices of $N$ are in the dominating set, we move the vertex $v$ to $W_l$, delete the vertices of $N\cap B_h$ and the edges from $v$ to $N\cap W_h$, to obtain the graph $G^\prime$. Starting from here, we then update all the sets in a way that the invariants are maintained. More precisely, set $B$ as the set of vertices dominated by $Y$, $W$ as the set of vertices not dominated by $Y$, set $B_h$ as the set of vertices of $B$ which dominate at least $d_b+1$ vertices of $W$, $B_l$ as the set of vertices of $B$ which dominate at most $d_b$ vertices of $W$, $W_h$ as the set of vertices of $W$ which have a neighbor in $B_h\cup W$, and $W_l$ as the rest of the vertices of $W$. Finally, we recurse on the resulting instance.\\

\noindent
{\bf Correctness.}
At each node of the recursion tree, we define a measure $\mu(I)=d_b (k-\vert Y\vert)-\vert W_l\vert$. We prove the correctness of the algorithm by induction on this measure. In the base case, when $d_b (k-\vert Y\vert)-\vert W_l\vert < 0$, then the algorithm is correct (by Lemma~\ref{lem:base_case_domset}). Now, we assume as induction hypothesis that the algorithm is correct on instances with measure less than some $\mu\geq 0$. Consider an instance $I$ such that $\mu(I)=\mu$. Since the branching is exhaustive, it is sufficient to show that the algorithm is correct on each of the child instances. To show this, it is sufficient to show that for each child instance $I^\prime$, $\mu(I^\prime)<\mu(I)$. In the first $d_w+1$ branches, the size of the set $Y$ increases by 1, and the size of the set $W_l$ does not decrease ($W_l$ has no neighbors in $B_h\cup W_l$). Hence, in each of these branches, $\mu(I^\prime)\leq \mu(I)-d_b$. In the final branch, though the size of the set $Y$ remains the 
same, the size of the set $W_l$ increases by at least 1. 
Hence, in this branch, $\mu(I^\prime)\leq \mu(I)-1$. Thus, we have shown that in each branch, the measure drops, hence completing the proof of correctness of the algorithm.\\

\begin{algorithm}[t]
   \SetKwInOut{Input}{Input}\SetKwInOut{Output}{Output}
  \Input{An instance $(G,k)$ of {\ds}, degree bound $d_b$, sets $B_h$, $B_l$, $Y$, $W_h$, $W_l$}
  \Output{A smallest solution of size at most $k$ and containing $Y$ for the instance $(D,k)$ if it exists and $S_\infty$ otherwise}

  \lIf{$\vert W_l\vert>d_b(k-\vert Y\vert)$}{\Return $S_\infty$}

   \ElseIf{$B_h=\emptyset$}{
      $S \leftarrow$ \textsc{Nederlof}$(G_{dst},r,W,W)$.\\
 \lIf{$|S| > k$}{$S \leftarrow S_\infty$}\\
 \Return $S$\\
 }

\Else{
$S \leftarrow S_\infty$\\
  \emph{Find vertex $v\in W_h$ with the least neighbors in $B_h\cup W_h$.}\\
\For{$u\in ((B_h \cup W_h)\cap N[v])$}{ $Y=Y\cup \{u\}$, perform updates to get $B_h^\prime$, $B_l^\prime$, $W_h^\prime$, $W_l^\prime$.\\
$S'\leftarrow$ {\sc DS-solve}($(G,k),d_b,B_h^\prime, B_l^\prime, Y, W_h^\prime, W_l^\prime$).\\
\lIf{$|S'| < |S|$}{$S \leftarrow S'$}\\
}
$W_l \leftarrow W_l\cup\{v\}$, perform updates to get new graph $G^\prime$ and sets $B_h^\prime$, $B_l^\prime$, $W_h^\prime$,$W_l^\prime$.\\
$S'\leftarrow$ {\sc DS-solve}($(G^\prime,k),d_b,B_h^\prime, B_l^\prime, Y, W_h^\prime, W_l$).\\
\lIf{$|S'| < |S|$}{$S \leftarrow S'$}\\
  \Return $S$
}
 \BlankLine
  \caption{Algorithm {\sc DS-solve} for {\ds}}\label{algo:minor_algo_domset}
\end{algorithm}
\noindent
{\bf Analysis.} 
Since $D$ exludes $K_h$ as a topological minor, Lemma~\ref{lem:minorfree_small_degree}, combined with the fact that we set $d_b = h-2$, implies that $d_w^{\max}=ch^4$, for some $c$, is an upper bound on the maximum $d_w$ which can appear during the execution of the algorithm. 
We first bound the number of leaves of the recursion tree as follows. The number of leaves is bounded by 
$\sum_{i=0}^{d_b k} \binom{d_b k}{i} (d_w^{\max}+1)^{k-{\frac i {d_b}}}$. 
To see this, observe that each branch of the recursion tree can be described by a length-$d_b k$ vector as shown in the correctness paragraph. We then select $i$ positions of this vector on which the last branch was taken. Finally for  $k-{\frac i {d_b}}$ of the remaining positions, we describe which of the first at most $(d_w^{\max}+1)$ branches was taken. Any of the first $d_w^{max}+1$ branches can be taken at most $k-{\frac i {d_b}}$ times if the last branch is taken $i$ times.

The time taken along each root to leaf path in the recursion tree is polynomial, while the time taken at a leaf  for which the last branch was taken $i$ times is $\bigoh^*(2^{d_b (k - (k-{\frac i {d_b}}))})=\bigoh^*(2^{i+k})$ 
(see Lemmata~\ref{lem:base_case_domset} and~\ref{lem:base_case_nederlof_domset}). 
Hence, the running time of the algorithm is 
\[
\bigoh^* \left(\sum_{i=0}^{d_bk} \binom{d_bk}{i} (d_w^{\max}+1)^{k-{\frac i {d_b}}}\cdot 2^{i+k} \right)
=\bigoh^*\left((2d_w^{\max}+2)^k \cdot \sum_{i=0}^{d_bk} \binom{d_bk}{i} \cdot 2^i \right) =\bigoh^*\left((2d_w^{\max}+2)^k \cdot 3^{d_b k} \right).
\]
For $d_b= h-2$ and $d_w^{\max} =ch^4$
 this is $\bigoh^*(3^{hk + o(hk)})$. This completes the proof of the theorem.

\end{proof}



\noindent
Observe that, when Algorithm \ref{algo:minor_algo_domset} is run on a graph of degeneracy $d$, setting $d_b=d$, combined with the simple fact that $d_w^{max}\leq d$ gives us the following theorem.

\begin{theorem}\label{}
 {\ds} can be solved in time $\bigoh^*(3^{dk+o(dk)})$ on graphs of degeneracy $d$.
\end{theorem}

\noindent
From the above two theorems and Lemma~\ref{lem:littleo}, we have the following corollary.

\begin{corollary}
If $\mathcal{C}$ is a class of graphs excluding $o(\log n)$-sized topological minors or an $o(\log n)$-degenerated class of graphs, then {\ds} parameterized by $k$ is {\FPT} on $\mathcal{C}$. 
\end{corollary}

\subsection{Hardness}

\begin{theorem}\label{thm:ds hardness logn}
{\ds} cannot be solved in time $f(k)n^{o({\frac {k} {\log k}})}$ on $c\log n$-degenerated graphs for any constant $c>0$, where $k$ is the solution size and $f$ is an arbitrary function, unless {\ETH} fails.
\end{theorem}

\begin{proof}
The proof is by a reduction from the restricted version of {\sc Set Cover} shown to be hard in Lemma~\ref{lem:set cover hardness}. Fix a constant $c>0$ and let $({\cal U}=\{u_1,\dots,u_n\},{\cal F}=\{F_1,\dots,F_m\},k)$ be an instance of {\sc Set Cover}, where the size of any set is at most $\gamma \log m$ for some constant $\gamma$. For each set $F_i$, we have a vertex $f_i$. For each element $u_i$, we have a vertex $x_i$. If an element $u_i$ is contained in the set $F_j$, then we add an edge $(f_j,x_i)$. Further, we add two vertices $r$ and $p$ and add edges $(r,f_i)$ for every $i$ and an edge $(r,p)$. Finally, we add star of size $m^{2\gamma/c}$ centered in $q$ disjoint from the rest of the graph. This completes the construction of the graph $G$. 

We claim that $({\cal U}, {\cal F}, k)$ is a {\Yes} instance of {\sc Set Cover} iff $(G,k+2)$ is a {\Yes} instance of {\ds}. Suppose that $\{F_1,\dots,F_k\}$ is a set cover for the given instance. It is easy to see that the vertices $q,r,f_1,\dots, f_k$ form a solution for the {\ds} instance.

In the converse direction, since one of $p$ and $r$ and at least one vertex of the star must be in any dominating set, we assume without loss of generality that $r$ and $q$ are contained in the minimum dominating set. Also, since $r$ dominates any vertex $f_i$, we may also assume that the solution is disjoint from the $x_i$'s. This is because, if $x_i$ was in the solution, we can replace it with an adjacent $f_j$ to get another solution of the same size. Hence, we suppose that $\{q,r,f_1,\dots, f_k\}$ is a solution for the {\ds} instance. Since the only way that some vertex $x_i$ can be dominated is by some $f_j$, and the construction implies that $u_i\in F_j$, the sets $\{F_1,\dots, F_k\}$ form a set cover for $({\cal U}, {\cal F}, k)$. This concludes the proof of equivalence of the two instances.

We claim that the degeneracy of the graph $G$ is bounded by $c \log  n_1$, where $n_1$ is the number of vertices in the graph $G$. First, we claim that the degeneracy of the graph $G$ is bounded by $\gamma \log m+1$. 
This follows from that each vertex $f_i$ has total degree at most $\gamma \log m +1$, each leaf of the star has degree 1 and if a subgraph contains none of these vertices, then it contains no edges.  Now, $n_1$ is at least $m^{2\gamma/c}$. Hence, $\log n_1  \geq (2\gamma/c)\log m$ and the degeneracy of the graph is at most $\gamma \log m +1 \leq c \cdot (2\gamma/c) \log m \leq c \log n_1$. Finally, since each vertex $f_i$ is incident to at most $\gamma \log m +1$ vertices, $n_1=\bigoh(m \log m + m^{2\gamma/c})$ and, thus, it is polynomial in $m$. Hence, an algorithm for {\ds} of the form $f(k)n_1^{o(\frac{k}{\log k})}$ implies an algorithm of the form $f(k)m^{o(\frac{k}{\log k})}$ for the {\sc Set Cover} instance. This concludes the proof of the theorem.
\end{proof}

\noindent
As a corollary of Theorem~\ref{thm:ds hardness logn}, and Lemma~\ref{lem:littleo}, we have the following corollary.


\begin{corollary}\label{cor:dependency on d}
There are no two functions $f$ and $g$ such that $g(d)=o(d)$ and there is an algorithm for {\ds} running in time $\bigoh^*(2^{g(d)f(k)})$ unless {\ETH} fails.
\end{corollary}


\noindent
From Theorem~\ref{thm:bounded degree subexp hard}, we can infer the following corollary.

\begin{corollary}\label{cor:dependency on k}
 There are no two functions $f$ and $g$ such that $f(k)=o(k)$ and there is an algorithm for {\ds} running in time $\bigoh^*(2^{g(d)f(k)})$, unless {\ETH} fails.
\end{corollary}

\begin{proof}
Suppose that there were an algorithm for {\ds} running in time $\bigoh^*(2^{g(d)f(k)})$, where $f(k)=o(k)$. Consider an instance $(G,k)$ of {\ds} where $G$ is a graph with maximum degree $c$. We can assume that the number of vertices of the graph is at most $ck+k$, since otherwise, it is a trivial {\No} instance. Hence, $k=\Theta(n)$ and thus, an algorithm running in time $\bigoh^*(2^{g(d)f(k)})$ will run in time $\bigoh^*(2^{g(c)f(k)})=\bigoh^*(2^{o(n)})$, and, by Theorem~\ref{thm:bounded degree subexp hard}, ETH fails.  
\end{proof}

\noindent
Thus, Corollaries \ref{cor:dependency on d} and \ref{cor:dependency on k} together show that, unless {\ETH} fails, our algorithm for {\ds} has the best possible dependence on both the degeneracy and the solution size.\\

\subsection{Approximating {\ds} on graphs of bounded degeneracy}

In this section, we adapt the ideas developed in the previous subsections, to design a polynomial-time $O(d^2)$-approximation algorithm for the {\ds} problem on $d$-degenerated graphs.

\begin{theorem}\label{thm:approx_ds}
 There is a $O(dn\log n)$-time $d^2$- approximation algorithm for the {\ds} problem on $d$-degenerated graphs. 
\end{theorem}

 \begin{proof}
The approximation algorithm is based on our {\FPT} algorithm, but whenever the branching algorithm would branch, we take all candidates into the solution and cycle instead of recursing. During the execution of the algorithm the partial solution is kept in the set $Y$ and vertex sets $B, W, B_h$, $B_l$, $W_h$, and $W_l$ are updated with the same meaning as in the branching algorithm. In the base case, all vertices of $W_l$ are taken into the solution. This last step could be replaced by searching a dominating set for the vertices in $W_l$ by some approximation algorithm for \textsc{Set Cover}. While this would probably improve the performance of the algorithm in practice, 
it does not 
improve the theoretical worst case bound.

\begin{algorithm}[t]
   \SetKwInOut{Input}{Input}\SetKwInOut{Output}{Output}
  \Input{An Undirected graph $G=(V,E)$ without isolated vertices}
  \Output{A dominating set for $G$ of size at most $O(d^2)$ times the size of a minimum dominating set}
$Y \leftarrow \emptyset$\\
$W_h \leftarrow V$\\
\While{$W_h \neq \emptyset$}
{\emph{Find vertex $v\in W_h$ with the least neighbors in $B_h\cup W_h$.}\label{approx_algo:find}\\
$Y \leftarrow Y \cup ((B_h \cup W_h)\cap N(v))$\label{approx_algo:branch}\\
$B \leftarrow$ vertices in $V\setminus Y$ with a neighbor in $Y$\\
$W \leftarrow V \setminus (Y \cup B)$\\
$B_h \leftarrow$ vertices in $B$ with at least $d+1$ neighbors in $W$\\
$B_l \leftarrow B \setminus B_h$\\
$W_h \leftarrow$ vertices in $W$ with a neighbor in $B_h$ or $W$\\
$W_l \leftarrow W \setminus W_h$\\
}
$Y \leftarrow Y \cup W_l$\label{approx_algo:base}\\
\Return $Y$\\
\BlankLine
\caption{Algorithm {\sc DS-approx} for {\ds}}\label{algo:approx_domset}
\end{algorithm}

 \noindent
{\bf Correctness.}
It is easy to see that the set $Y$ output by the algorithm is a dominating set for $G$. Now let $Q$ be an optimal dominating set for $G$. We want to show that $|Y|$ is at most $d^2$ times $|Q|$. In particular, we want to account every vertex of $Y$ to some vertex of $Q$, which ``should have been chosen instead to get the optimal set.'' Before we do that let us first observe, how the vertices can move around the sets during the execution of the algorithm. Once a vertex is added to $Y$ it is never removed. Hence, once a vertex is moved from $W$ to $B$, it is never moved back. But then the vertices in $B$ are only losing neighbors in $W$, and once they get to $B_l$ they are never moved anywhere else. Thus, vertices in $W_l$ can never get a new neighbor in $W$ or $B_h$ and they also stay in $W_l$ until Step~\ref{approx_algo:base}. The vertices of $B_h$ can get to $Y$ or $B_l$ and the vertices of $W_h$ can get to any other set during the execution of the algorithm.

Now let $v$ be the vertex found in Step~\ref{approx_algo:find} of Algorithm~\ref{algo:approx_domset} and $w$ be the vertex dominating it in $Q$. As in the branching algorithm, we know that $|(B_h \cup W_h)\cap N(v)| \leq d$ since the subgraph of $G$ induced by $(B_h \cup W_h)$ is $d$-degenerated. Hence at most $d$ vertices are added to $Y$ in Step~\ref{approx_algo:branch} of the algorithm. We charge the vertex $w$ for these at most $d$ vertices and make the vertex $v$ responsible for that. Finally, in Step~\ref{approx_algo:base} let $v \in W_l$ and $w$ be a vertex which dominates it in $Q$. We charge $w$ for adding $v$ to $Y$ and make $v$ responsible for it. Obviously, for each vertex added to $Y$, some vertex is responsible and some vertex of $Q$ is charged. It remains to count for how many vertices can be a vertex of $Q$ charged. 

Observe that whenever a vertex is responsible for adding some vertices to $Y$, it is $W$ and after adding these vertices it becomes dominated and, hence, moved to $B$. Therefore, each vertex becomes responsible for adding vertices only once and, thus, each vertex is responsible for adding at most $d$ vertices to $Y$. Now let us distinguish in which set a vertex $w$ of $Q$ is, when it is first charged. If $w$ is first charged in Step~\ref{approx_algo:branch} of the algorithm, then a vertex $v \in W_h$ is responsible for that and $w$ has to be either in $W_h$, $B_h$, or $B_l$, as vertices in $W_h$ do not have neighbors elsewhere. In the first two cases, if $w \neq v$, then it is moved to $Y$ and hence has no neighbors in $W$ anymore. If $v=w$, then it is moved to $B_l$, and it has no neighbors in $W$, as all of them are moved to $Y$. Thus, in these cases, after the step is done, $w$ has no neighbors in $W$ and, hence, is never charged again. If $w$ is in $B_l$, then it has at most $d$ neighbors in $W$ and, as each of them is responsible for adding at most $d$ vertices to $Y$, $w$ is charged for at most $d^2$ vertices. If a vertex $w$ is first charged in Step~\ref{approx_algo:base}, then it is $B_l$ or $W_l$, has at most $d$ neighbors in $W_l$, each of them being responsible for addition of exactly one vertex, so it is charged for addition of at most $d$ vertices. It follows, that every vertex of $Q$ is charged for addition of at most $d^2$ to $Y$ and, therefore, $|Y|$ is at most $d^2$ times $|Q|$.

 \noindent
{\bf Running time analysis.} 
To see the running time, observe first that by the above argument, every vertex is added to each of the sets at most once. Also a vertex is moved from one set to another only if some of its neighbors is moved to some other set or it is selected in Step~\ref{approx_algo:find}. Hence, we might think of a vertex sending a signal to all its neighbor, once it is moved to another set. There are only constantly many signals and each of them is sent at most once over each edge in each direction. Hence, all the updations of the sets can be done in $O(m)=O(dn)$-time, as the graph is $d$-degenerate. 
Also each vertex in $W_h$ can keep its number of neighbors in $W_h \cup B_h$ and update it whenever it receives a signal from some of its neighbors about being moved out of $W_h \cup B_h$. We can keep a heap of the vertices in $W_h$ sorted by a degree in  $W_h \cup B_h$ and update it in $O(\log n)$ time whenever the degree of some of the vertices change. This means $O(dn\log n)$ time to keep the heap through the algorithm. Using the heap, the vertex $v$ in Step~\ref{approx_algo:find} can be found in $O(\log n)$-time in each iteration. As in each iteration at least one vertex is added to $Y$, there are at most $n$ iterations and the total running time is $O(dn\log n)$.
\end{proof}

This algorithm can be also used for $K_h$-minor-free and $K_h$-topological-minor free graphs, yielding $O(h^2\cdot \log h)$-approximation and $O(h^4)$-approximation, as these graphs are $O(h \cdot \sqrt{\log h})$ and $O(h^2)$-degenerated, respectively. As far as we know, this is also the first constant factor approximation for dominating set in $K_h$-topological-minor free graphs. Although PTAS is known for $K_h$-minor-free graphs~\cite{Grohe03}, our algorithm can be still of interest due to its simplicity and competitive running time.

%% file: conclusion.tex
We gave the first {\FPT} algorithms for the {\sc Steiner Tree} problem on directed graphs excluding a fixed graph as a (topological) minor, and then extended the results to directed graphs of bounded degeneracy. We mention that the same approach also gives us {\FPT} algorithms for {\ds} and some of its variants, for instance {\sc Connected} {\ds} and {\sc Total} {\ds}. Finally, in the process of showing the optimality of our algorithm, we showed that for any constant $c$, {\dst} is not expected to have an algorithm of the form $f(k)n^{o(\frac{k}{\log k})}$ on $o(\log n)$-degenerated graphs. It would be interesting to either improve this lower bound, or prove the tightness of this bound by giving an algorithm with a matching running time.

%% file: steinermasterfile.bbl
\begin{thebibliography}{10}

\bibitem{AlberBFKN02}
J.~Alber, H.~L. Bodlaender, H.~Fernau, T.~Kloks, and R.~Niedermeier.
\newblock Fixed parameter algorithms for dominating set and related problems on
  planar graphs.
\newblock {\em Algorithmica}, 33(4):461--493, 2002.

\bibitem{AlonG09}
N.~Alon and S.~Gutner.
\newblock Linear time algorithms for finding a dominating set of fixed size in
  degenerated graphs.
\newblock {\em Algorithmica}, 54(4):544--556, 2009.

\bibitem{BernPlassmann89ipl}
M.~W. Bern and P.~E. Plassmann.
\newblock The {S}teiner problem with edge lengths 1 and 2.
\newblock {\em Inf. Process. Lett.}, 32(4):171--176, 1989.

\bibitem{BjorklundHKK07}
A.~Bj{\"o}rklund, T.~Husfeldt, P.~Kaski, and M.~Koivisto.
\newblock Fourier meets {M}{\"o}bius: fast subset convolution.
\newblock In {\em Proc. STOC 2007}, pages 67--74, 2007.

\bibitem{BollobasBT98}
B.~Bollob\'{a}s and A.~Thomason.
\newblock Proof of a conjecture of {M}ader, {E}rd\"{o}s and {H}ajnal on
  topological complete subgraphs.
\newblock {\em Eur. J. Comb.}, 19(9):883--887, 1998.

\bibitem{CharikarCCDGGL99}
M.~Charikar, C.~Chekuri, T.-Y. Cheung, Z.~Dai, A.~Goel, S.~Guha, and M.~Li.
\newblock Approximation algorithms for directed {S}teiner problems.
\newblock {\em J. Algorithms}, 33(1):73--91, 1999.

\bibitem{DuhF97}
R.~chii Duh and M.~F{\"u}rer.
\newblock Approximation of {\it k}-set cover by semi-local optimization.
\newblock In F.~T. Leighton and P.~W. Shor, editors, {\em STOC}, pages
  256--264. ACM, 1997.

\bibitem{ChlebikC08}
M.~Chleb\'{\i}k and J.~Chleb\'{\i}kov\'{a}.
\newblock Approximation hardness of dominating set problems in bounded degree
  graphs.
\newblock {\em Inf. Comput.}, 206(11):1264--1275, Nov. 2008.

\bibitem{CyganNPPRW11}
M.~Cygan, J.~Nederlof, M.~Pilipczuk, M.~Pilipczuk, J.~M.~M. van Rooij, and
  J.~O. Wojtaszczyk.
\newblock Solving connectivity problems parameterized by treewidth in single
  exponential time.
\newblock In {\em Proc. FOCS 2011}, pages 150--159, 2011.

\bibitem{DemaineFHT05sub}
E.~D. Demaine, F.~V. Fomin, M.~Hajiaghayi, and D.~M. Thilikos.
\newblock Subexponential parameterized algorithms on bounded-genus graphs and
  {$H$}-minor-free graphs.
\newblock {\em J. ACM}, 52(6):866--893, 2005.

\bibitem{DemaineHaj05}
E.~D. Demaine and M.~Hajiaghayi.
\newblock Bidimensionality: new connections between {FPT} algorithms and
  {PTAS}s.
\newblock In {\em Proc. SODA 2005}, pages 590--601. ACM-SIAM, 2005.

\bibitem{DemaineHK09a}
E.~D. Demaine, M.~Hajiaghayi, and P.~N. Klein.
\newblock Node-weighted {S}teiner tree and group {S}teiner tree in planar
  graphs.
\newblock In {\em Proc. ICALP 2009}, pages 328--340, 2009.

\bibitem{DF99}
R.~G. Downey and M.~R. Fellows.
\newblock {\em Parameterized Complexity}.
\newblock Springer-Verlag, New York, 1999.

\bibitem{DreyfusWagner72networks}
S.~E. Dreyfus and R.~A. Wagner.
\newblock The {S}teiner problem in graphs.
\newblock {\em Networks}, 1(3):195--207, 1971.

\bibitem{Eppstein00}
D.~Eppstein.
\newblock Diameter and treewidth in minor-closed graph families.
\newblock {\em Algorithmica}, 27(3):275--291, 2000.

\bibitem{FG06}
J.~Flum and M.~Grohe.
\newblock {\em Parameterized Complexity Theory}.
\newblock Springer-Verlag, Berlin, 2006.

\bibitem{FominT06}
F.~V. Fomin and D.~M. Thilikos.
\newblock Dominating sets in planar graphs: Branch-width and exponential
  speed-up.
\newblock {\em SIAM J. Comput.}, 36:281--309, 2006.

\bibitem{FuchsKMRRW07}
B.~Fuchs, W.~Kern, D.~M{\"o}lle, S.~Richter, P.~Rossmanith, and X.~Wang.
\newblock Dynamic programming for minimum {S}teiner trees.
\newblock {\em Theory Comput. Syst.}, 41(3):493--500, 2007.

\bibitem{GareyJ77}
M.~R. Garey and D.~S. Johnson.
\newblock The rectilinear {S}teiner tree problem is {NP}-complete.
\newblock {\em SIAM J. Appl. Math.}, 32(4):826--834, 1977.

\bibitem{GareyJohnson79}
M.~R. Garey and D.~S. Johnson.
\newblock {\em Computers and Intractability}.
\newblock Freeman, San Francisco, 1979.

\bibitem{Grohe03}
M.~Grohe.
\newblock Local tree-width, excluded minors, and approximation algorithms.
\newblock {\em Combinatorica}, 23:613--632, 2003.
\newblock 10.1007/s00493-003-0037-9.

\bibitem{GroheM12}
M.~Grohe and D.~Marx.
\newblock Structure theorem and isomorphism test for graphs with excluded
  topological subgraphs.
\newblock {\em To appear, STOC}, 2012.

\bibitem{GuhaK98}
S.~Guha and S.~Khuller.
\newblock Approximation algorithms for connected dominating sets.
\newblock {\em Algorithmica}, 20(4):374--387, 1998.

\bibitem{GuoNS11}
J.~Guo, R.~Niedermeier, and O.~Such{\'{y}}.
\newblock Parameterized complexity of arc-weighted directed steiner problems.
\newblock {\em SIAM Journal on Discrete Mathematics}, 25(2):583--599, 2011.

\bibitem{HalperinKKSW07}
E.~Halperin, G.~Kortsarz, R.~Krauthgamer, A.~Srinivasan, and N.~Wang.
\newblock Integrality ratio for group {S}teiner trees and directed {S}teiner
  trees.
\newblock {\em SIAM J. Comput.}, 36(5):1494--1511, 2007.

\bibitem{HwangRichardsWinter92}
F.~K. Hwang, D.~S. Richards, and P.~Winter.
\newblock {\em The {S}teiner Tree Problem}.
\newblock North-Holland, Amsterdam, 1992.

\bibitem{ImpagliazzoPZ01}
R.~Impagliazzo, R.~Paturi, and F.~Zane.
\newblock Which problems have strongly exponential complexity?
\newblock {\em J. Comput. Syst. Sci.}, 63(4):512--530, 2001.

\bibitem{KahngRobins95}
A.~B. Kahng and G.~Robins.
\newblock {\em On Optimal Interconnections for {VLSI}}.
\newblock Kluwer Academic Publisher, 1995.

\bibitem{KleinR95}
P.~Klein and R.~Ravi.
\newblock A nearly best-possible approximation algorithm for node-weighted
  steiner trees.
\newblock {\em Journal of Algorithms}, 19(1):104 -- 115, 1995.

\bibitem{KomloS96}
J.~Koml{\'o}s and E.~Szemer{\'e}di.
\newblock Topological cliques in graphs 2.
\newblock {\em Combinatorics, Probability {\&} Computing}, 5:79--90, 1996.

\bibitem{KortePromelSteger90}
B.~Korte, H.~J. Pr{\"o}mel, and A.~Steger.
\newblock Steiner trees in {VLSI}-layout.
\newblock In {\em Paths, Flows and VLSI-Layout}, pages 185--214, 1990.

\bibitem{Marx07}
D.~Marx.
\newblock Can you beat treewidth?
\newblock {\em Theory of Computing}, 6(1):85--112, 2010.

\bibitem{MisraPRSS10}
N.~Misra, G.~Philip, V.~Raman, S.~Saurabh, and S.~Sikdar.
\newblock {FPT} algorithms for connected feedback vertex set.
\newblock In {\em Proc. WALCOM 2010}, pages 269--280, 2010.

\bibitem{MolleRR08}
D.~M{\"o}lle, S.~Richter, and P.~Rossmanith.
\newblock Enumerate and expand: Improved algorithms for connected vertex cover
  and tree cover.
\newblock {\em Theory Comput. Syst.}, 43(2):234--253, 2008.

\bibitem{Nederlof09}
J.~Nederlof.
\newblock Fast polynomial-space algorithms using {M}{\"o}bius inversion:
  Improving on {S}teiner tree and related problems.
\newblock In {\em Proc. ICALP 2009}, pages 713--725, 2009.

\bibitem{Nie06}
R.~Niedermeier.
\newblock {\em Invitation to Fixed-Parameter Algorithms}.
\newblock Oxford University Press, Oxford, 2006.

\bibitem{PhilipRS09}
G.~Philip, V.~Raman, and S.~Sikdar.
\newblock Solving dominating set in larger classes of graphs: {FPT} algorithms
  and polynomial kernels.
\newblock In {\em Proc. ESA 2009}, volume 5757 of {\em LNCS}, pages 694--705.
  Springer, 2009.

\bibitem{PromelS02}
H.~J. Pr{\"o}mel and A.~Steger.
\newblock {\em The {S}teiner Tree Problem; a Tour through Graphs, Algorithms,
  and Complexity}.
\newblock Vieweg, 2002.

\bibitem{Zelikovsky97}
A.~Zelikovsky.
\newblock A series of approximation algorithms for the acyclic {D}irected
  steiner tree problem.
\newblock {\em Algorithmica}, 18(1):99--110, 1997.

\bibitem{ZosinK02}
L.~Zosin and S.~Khuller.
\newblock On directed {S}teiner trees.
\newblock In {\em Proc. SODA 2002}, pages 59--63, 2002.

\end{thebibliography}
